\documentclass[11pt,letterpaper]{article}

\usepackage{amsmath,amsthm,amsfonts,amssymb}
\usepackage{thm-restate}
\usepackage{fullpage}
\usepackage[utf8]{inputenc}
\usepackage[dvipsnames]{xcolor}
\usepackage{xspace,enumerate}
\usepackage[shortlabels]{enumitem}
\usepackage[hypertexnames=false,colorlinks=true,urlcolor=Blue,citecolor=Green,linkcolor=BrickRed]{hyperref}
\usepackage[capitalise]{cleveref}
\usepackage[OT4]{fontenc}
\usepackage[ruled]{algorithm}
\usepackage[noend]{algorithmic}
\usepackage{ifpdf}
\usepackage{todonotes}
\usepackage{enumerate}
\usepackage{authblk}
\usepackage{thmtools}

\usepackage[margin=1in]{geometry}
\def\poly{\operatorname{poly}}

\title{Subquadratic Dynamic Path Reporting in Directed Graphs\\Against an Adaptive Adversary\thanks{This work has been partially supported by the ERC CoG grant TUgbOAT no 772346.}}

\date{\vspace{-5ex}}
\author[1]{Adam Karczmarz\thanks{\texttt{a.karczmarz@mimuw.edu.pl}}}
\author[1]{Anish Mukherjee\thanks{\texttt{anish@mimuw.edu.pl}}}
\author[1]{Piotr Sankowski\thanks{\texttt{sank@mimuw.edu.pl}}}
\affil[1]{Institute of Informatics, University of Warsaw and IDEAS NCBR, Poland}

\newcommand{\Ot}{\ensuremath{\widetilde{O}}}
\newcommand{\eps}{\ensuremath{\epsilon}}

\newcommand{\wei}{w}

\theoremstyle{plain}
\newtheorem{theorem}{Theorem}[section]
\newtheorem{lemma}[theorem]{Lemma}
\newtheorem{corollary}[theorem]{Corollary}

\newtheorem{observation}[theorem]{Observation}

\newtheorem{remark}[theorem]{Remark}

\begin{document}

\maketitle
\thispagestyle{empty}

\begin{abstract}
  We study reachability and shortest paths problems in dynamic directed graphs.
  Whereas algebraic dynamic data structures supporting edge updates and reachability/distance queries
  have been known for quite a long time, they do not, in general, allow
  reporting the underlying paths within the same time bounds, especially against an adaptive adversary.

  In this paper we develop the first known fully dynamic reachability data structures working against an adaptive adversary and supporting edge updates and
  path queries for two natural variants: (1) point-to-point path reporting, and (2) single-source reachability tree
  reporting. 
  For point-to-point queries in DAGs, we achieve $O(n^{1.529})$ worst-case update and query bounds, whereas
  for tree reporting in DAGs, the respective worst-case bounds are $O(n^{1.765})$.
  More importantly, we show how to lift these algorithms to work on general graphs at the cost
  of increasing the bounds to $n^{1+5/6+o(1)}$ and making the update times amortized.
  On the way to accomplishing these goals, we obtain two interesting subresults. We give subquadratic fully dynamic algorithms
  for topological order (in a DAG), and strongly connected components. To the best of our knowledge, such algorithms
  have not been described before.

  Additionally, we provide \emph{deterministic} \emph{incremental} data structures for (point-to-point
  or single-source) reachability and shortest paths that can handle edge insertions and report the respective
  paths within subquadratic \emph{worst-case} time bounds.
  For reachability and \linebreak $(1+\eps)$-approximate shortest paths in weighted directed graphs, these bounds match the best known
  dynamic matrix inverse-based randomized bounds for fully dynamic\linebreak reachability~\cite{BrandNS19}.
  For exact shortest paths in unweighted graphs, the obtained bounds in the incremental setting
  polynomially improve upon the respective best known randomized update/distance query bounds in the fully dynamic setting.

\end{abstract}

\clearpage
\setcounter{page}{1}

\section{Introduction}\label{s:intro}

Dynamic reachability (or transitive closure) and dynamic all-pairs shortest paths are among the most fundamental
and well-studied dynamic problems on directed graphs.
In these problems, we are given a dynamic directed graph $G=(V,E)$ with $n=|V|$,
and the goal is to devise a data structure maintaining~$G$ and supporting edge set updates interleaved
with reachability (or shortest path) queries between arbitrary source-target pairs of vertices of $G$.
A dynamic data structure is called \emph{incremental} if it can handle edge insertions only,
\emph{decremental} if it can handle edge deletions only, and \emph{fully dynamic} if it can handle both.
One is typically interested in optimizing both the (amortized or worst-case) update time
of the data structure, and the query time (which is usually worst-case anyway).
In the partially dynamic settings one usually optimizes the total update time, i.e.,
the time needed to process the entire sequence of updates.

Efficient and non-trivial combinatorial\footnote{That is, not using algebraic methods like the Zippel-Schwartz lemma \cite{Zippel79, Schwartz80} or fast matrix multiplication.} algorithms have been developed for dynamic transitive closure in the incremental~\cite{Italiano86},
decremental~\cite{FrigioniMNZ01}, and fully dynamic settings~\cite{DemetrescuI08, Roditty08, RodittyZ16}.
For all-pairs shortest paths, Demetrescu and Italiano~\cite{DemetrescuI04}~developed a fully dynamic algorithm (later slightly
improved by~\cite{Thorup04}) with
$\Ot(n^2)$ amortized update time and optimal query time. That algorithm works in the most general
comparison-addition model which allows real edge weights.
Partially dynamic algorithms for all-pairs shortest paths with non-trivial total update
time bounds and optimal query time are known in the incremental~\cite{AusielloIMN91, Bernstein16, KarczmarzL19, KarczmarzL20}
and decremental~\cite{Bernstein16,EvaldFGW21,KarczmarzL20} settings; however, these algorithms work only for unweighted
graphs (or for small-integer weights) and/or produce $(1+\eps)$-approximate answers.

One can observe the following phenomenon: the known combinatorial dynamic algorithms for reachability and shortest paths
typically have good amortized, but much worse worst-case update bounds.
Moreover, in the fully dynamic settings, none of the state-of-the-art combinatorial algorithms~\cite{DemetrescuI04, DemetrescuI08, Roditty08, RodittyZ16, Thorup04}
achieves $O(n^{2-\delta})$ update time and query time at the same time for dense graphs (for any $\delta>0$).\footnote{Abboud and Vassilevska Williams~\cite{AbboudW14} explained why this may be impossible without resorting to fast matrix multiplication.}
King and Sagert~\cite{KingS02} were the first to observe that path counting modulo a sufficiently
large prime can be used to obtain a fully dynamic transitive closure data structure with $O(n^2)$ \emph{worst-case} update time
in the case of acyclic graphs.
Demetrescu and Italiano~\cite{DemetrescuI05} combined this technique with fast rectangular matrix multiplication
and obtained a fully dynamic reachability algorithm for DAGs supporting single-edge updates
in $O(n^{1+\rho})$ worst-case time and queries in $O(n^{\rho})$ time, where $\rho\approx 0.529$
equals the smallest real number such that $\omega(1,\rho,1)=1+2\rho$~\cite{GallU18}.
Here, $\omega(1,\alpha,1)$ denotes an exponent such that a product of $n\times n^{\alpha}$
and $n^\alpha\times n$ matrices can be computed in $O(n^{\omega(1,\alpha,1)})$ time.
Finally, Sankowski~\cite{Sankowski04} obtained the same bounds for fully dynamic reachability
in general graphs by reducing the problem to dynamically maintaining a matrix inverse.
This technique also led to subquadratic fully dynamic algorithms for exact distances in unweighted
graphs~\cite{Sankowski05}, approximate distances in weighted graphs~\cite{BrandN19}, and maximum matchings~\cite{Sankowski07}.
The state-of-the-art bounds for dynamic matrix inverse (and thus also for some of these graph problems)
were given by~\cite{BrandNS19}; we refer to that work for more graph applications
of dynamic matrix inverse.

The common drawback of algebraic dynamic transitive closure and distances algorithms based
on either dynamic path counting or dynamic matrix inverse is that they are Monte Carlo
randomized (as they invoke the Zippel-Schwartz lemma).\footnote{Recall we focus on directed graphs throughout. Interestingly, in a recent manuscript,~\cite{BrandFN21} show a \emph{deterministic} fully dynamic algorithm for approximating distances in \emph{undirected} graphs that uses algebraic techniques.} More importantly, they do not, in general, allow constructing a path that certifies reachability or achieves the reported distance.
To the best of our knowledge, the only exception to the latter drawback is the recent trade-off algorithm of~\cite{BergamaschiHGWW21}
which, if the maximum of query/update time bounds is optimized,
supports single-edge updates and shortest path reporting queries in $O(n^{1.897})$ worst-case
time.\footnote{In \cite{BergamaschiHGWW21}, the description of this algorithm is deferred to the full version~\cite[Section~5.4]{BergamaschiHGWW21full}.}
However, the high probability correctness of the dynamic shortest paths data structure of~\cite{BergamaschiHGWW21}
is only guaranteed against an \emph{oblivious adversary}
that does not base its future updates on the answers that the data structure produces.\footnote{More specifically,
the reported paths in~\cite{BergamaschiHGWW21} can leak randomness of two different randomized
components: the hitting sets of sufficiently long paths, and the witnesses of maintained
matrix products.}
In the recent years, there has been a significant effort of the dynamic graph algorithms
community to obtain solutions that work well against an \emph{adaptive adversary},
possibly by completely avoiding randomization (see, e.g.,~\cite{BernsteinGS20, ChuzhoyK19, EvaldFGW21, NanongkaiS17}).
  Such algorithms are not only more general, but can also be used in a black-box
  way as building blocks for static algorithms.

  To the best of our knowledge, none of the currently known dynamic algorithms for reachability
  or shortest paths that relies on algebraic techniques and has subquadratic update/query bounds is able to report the certifying
  paths and still perform well under the adaptive adversary assumption.

  \subsection{Our results}

  \subsubsection{Fully dynamic path reporting}

  As our first contribution, we show the first known fully dynamic reachability
  algorithms that support edge updates and certificate reporting queries in subquadratic time
  and at the same time work against an adaptive adversary.
  We consider two natural certificate reporting variants:
  \begin{enumerate}[label=(\arabic*)]
    \item point-to-point path reporting for any requested source-target
  pair $s,t\in V$, 
  and,
    \item single-source reachability tree reporting for any requested source $s\in V$.
  \end{enumerate}

  Note that the former variant constitutes a reporting analogue to reachability (transitive closure) queries,
  whereas the latter is an analogue to single-source reachability queries. Recall
  that $\rho\approx 0.529$ is such that $\omega(1,\rho,1)=1+2\rho$. Moreover, if $\omega=\omega(1,1,1)=2$, then $\rho=0.5$.
  The following theorem summarizes the bounds that we obtain for directed acyclic graphs.

\begin{theorem}\label{t:dags}
  Let $G$ be a directed graph subject to fully dynamic single-edge updates that keep~$G$
  acyclic at all times. Then, there exist data structures with the following \emph{worst-case} update and query bounds:
  \begin{enumerate}[label=(\arabic*)]
    \item $\Ot(n^{1+\rho})=O(n^{1.529})$ for point-to-point path reporting,
    \item $\Ot(n^{(3+\rho)/2})=O(n^{1.765})$ for single-source reachability tree reporting.
  \end{enumerate}
  The data structures are Monte Carlo randomized and produce
  answers correct with high probability against an adaptive adversary.
\end{theorem}

In order to obtain the former of the above data structures, we observe that a topological order
of a fully dynamic acyclic graph can be maintained in $\Ot(n^{1+\rho})=O(n^{1.529})$ worst-case
time per update as well. Whereas this observation is a simple consequence of
the known dynamic reachability algorithms with subquadratic update time and sublinear query time~\cite{DemetrescuI05,Sankowski04},
to the best of our knowledge, it has not been described before.
Dynamic topological ordering has been mostly studied in the incremental setting,
and multiple algorithms with non-trivial total update time bounds are known~\cite{BenderFGT16, BernsteinC18, BhattacharyaK20}.

\newcommand{\sccs}{\mathcal{S}}
\newcommand{\decscc}{\mathcal{C}}

Even more importantly, we show the following reductions of the respective variants
on general graphs to suitable reporting variants of the \emph{decremental strongly connected components}
problem.
In this paper, we use the term \emph{strongly connected components} to refer to the set $\sccs$ of equivalence classes of the strong connectivity relation,
and not the subgraphs $G[S]$, $S\in \sccs$

\begin{restatable}{theorem}{tpathreporting}\label{t:path-reporting}
  Let $\decscc$ be a decremental data structure with total update time\footnote{Throughout, we assume that the total update time also captures the initialization time. Moreover, for decremental data structures, $m$ denotes the initial number of edges of the data structure's input graph.} ${T(n,m)=\Omega(m+n)}$
  (1) maintaining strongly connected components  explicitly and (2) supporting
  queries reporting a simple path $P$ between arbitrary \emph{strongly connected} vertices $u,v$
  of the maintained graph in \linebreak $O(|P|\cdot n^{\rho})$ time.
  Then, there exists a fully dynamic data structure supporting single-edge updates and
  point-to-point path reporting queries
  with amortized update time and worst-case query time of $\Ot\left(\sqrt{T(n,m)\cdot n}+n^{1+\rho}\right)$.

  If the data structure $\decscc$ works against an adaptive adversary,
  so does the fully dynamic path reporting data structure.
\end{restatable}
\begin{restatable}{theorem}{ttreereporting}\label{t:tree-reporting}
  Let $\decscc$ be a decremental data structure with total update time $T{(n,m)=\Omega(m+n)}$\linebreak
  (1) maintaining strongly connected components $\sccs$ of $G$ explicitly and (2) supporting
  queries reporting, for a chosen strongly connected component $S\in\sccs$,
  a (possibly sparser) strongly connected subgraph $Z\subseteq G[S]$ with $V(Z)=S$ 
  in $O(|S|\cdot n^{(1+\rho)/2})$ time. Then, there exists a fully dynamic data structure supporting
  single-edge updates and single-source reachability tree reporting queries
  with amortized update time and worst-case query time of $\Ot\left(\sqrt{T(n,m)\cdot n}+n^{(3+\rho)/2}\right)$.

  If the data structure $\decscc$ works against an adaptive adversary,
  so does the fully dynamic reachability tree reporting data structure.
\end{restatable}

In particular, a deterministic decremental strongly connected components data structure with total update time $T(n,m)=mn^{2/3+o(1)}$ that satisfies the requirements
of both reductions has been recently shown by~\cite{BernsteinGS20}, i.e., the query time of their data structure is $|P|\cdot n^{o(1)}=$ \linebreak $O(|P|\cdot n^{\rho})$ for path reporting and $|S|\cdot n^{o(1)}$ for strongly connected subgraph reporting, respectively.
Hence, we obtain $n^{1+5/6+o(1)}$  \emph{amortized} update bound and worst-case query time
for both variants against an adaptive adversary.
Thanks to our reductions, further progress on deterministic (or adaptive\footnote{That is, allowing path-reporting queries within a strongly connected component against an adaptive adversary.}) decremental strongly connected components
problem will lead to improved bounds for subquadratic fully dynamic path- and tree-reporting
data structures.

We also note that if the oblivious adversary assumption is acceptable,
by plugging in the near-optimal decremental strongly connected components
data structure of~\cite{BernsteinPW19} with $T(n,m)=\Ot(n+m)$ as $\decscc$,
we can obtain the same asymptotic (but still amortized) bounds for the respective reporting
variants as in Theorem~\ref{t:dags} for acyclic graphs.

As a warm-up to proving Theorems~\ref{t:path-reporting}~and~\ref{t:tree-reporting},
we also show that the strongly connected components of a \emph{fully dynamic}
graph can be explicitly maintained under single-edge updates in $\Ot(n^{1+\rho})$ \emph{amortized}
time per update. To the best of our knowledge, no non-trivial bounds
have ever been described for the strongly connected components problem
in the fully dynamic setting for general digraphs.
Such a bound (in fact, a worst-case bound) has only been given for an easier problem of testing whether
the graph \emph{is strongly connected}~\cite{BrandNS19}.
It is known that one cannot achieve a subquadratic update bound for this problem
using combinatorial methods~\cite{AbboudW14}.

\subsubsection{Deterministic incremental algorithms with subquadratic worst-case bounds}
As our second contribution, we show that randomization is not always required for obtaining
subquadratic \emph{worst-case} update bounds for dynamic reachability or shortest paths problems in directed graphs.\footnote{See~\cite{BrandFN21} for subquadratic \emph{deterministic} fully dynamic approximate distances algorithms in \emph{unweighted undirected} graphs.}
Namely, we show that in the \emph{incremental} setting, there exist \emph{deterministic} \emph{path-reporting} algebraic
data structures for reachability, and (approximate) shortest paths.
These data structures completely avoid using the Zippel-Schwartz lemma.

Whereas in the incremental setting one usually studies the total update time,
obtaining incremental data structures with good worst-case bounds is important for the following
additional reasons. First of all, such data structures can efficiently handle \emph{rollbacks}, i.e.,
are capable of reverting the most recent insertion within the same worst-case time bound.
As a result, they are useful in certain limited fully dynamic settings as well.
This property can be also used to obtain \emph{offline} fully dynamic algorithms\footnote{That is, in the case when the entire sequence of updates issued is known beforehand.} with the same update bound in a black-box way (up to polylogarithmic factors).
Formally, we have the following transformation\footnote{Technically, the transformation requires making the incremental data structure fully persistent first. However, this can be easily achieved deterministically and without introducing additional amortization using standard methods~\cite{Dietz89a, DriscollSST89}, at the cost of only a polylogarithmic slowdown of updates and queries.} (see, e.g.~\cite[Theorem 1]{LackiS13}).
\begin{lemma}\label{l:offline}
Suppose there is an \emph{incremental} data structure maintaining some information about
the graph $G$ with initialization time $I(n,m)$, worst-case update time $U(n,m)$, and query time $Q(n,m)$.
Then, one can preprocess a sequence of $t$ \emph{fully dynamic} updates to $G$ (given offline) in $\Ot(I(n,m)+t\cdot U(n,m))$ time,
so that queries about any of the $t+1$ versions of the graph are supported in $\Ot(Q(n,m))$ time.
The transformation is deterministic.
\end{lemma}
Offline algorithms are, in turn, important from the hardness perspective -- many of the known conditional
lower bound techniques for dynamic problems (see, e.g,~\cite{AbboudW14, GutenbergWW20}) apply to the offline setting as well,
and thus obtaining a faster offline algorithm can exclude the possibility that a certain conditional lower bound exists.

For incremental reachability, we show the following.
\begin{theorem}\label{t:incremental-reach-sum}
  There exist deterministic incremental reachability data structures supporting:
  \begin{itemize}
    \item single-edge insertions in $\Ot(n^{1+\rho})=O(n^{1.529})$ worst-case time and path-reporting queries in $\Ot(n^{\rho})+O(|P|)=O(n^{0.529}+|P|)$ time,
      where $P$ is the reported simple path,\footnote{We stress that if one only cares about path existence, and not a certificate,
      the $|P|$ term can be omitted. This also applies to other results stated in this section.}
    \item insertions of at most $n$ \emph{incoming edges} of a single vertex, and single-source reachability tree
      queries, both in $O(n^{1.529})$ worst-case time.
  \end{itemize}
\end{theorem}
The bounds in Theorem~\ref{t:incremental-reach-sum} match the best-known dynamic matrix inverse-based bounds for \emph{fully dynamic} transitive closure
in the respective update/query variants (that is, single-edge updates/single-pair queries, or incoming edges updates/single-source queries, resp.)~\cite{BrandNS19}.
However, as Theorem~\ref{t:incremental-reach-sum} shows, in the incremental (or offline fully dynamic, by Lemma~\ref{l:offline}) setting, one can reproduce these bounds
deterministically and allow for very efficient certificate reporting.

For incremental $(1+\eps)$-approximate shortest paths, we show the following.
\begin{theorem}\label{t:incremental-appr-sum}
  Let $\eps\in (0,1]$ and suppose $G$ is a weighted digraph with real edge weights in $[1,C]$.
  There exist deterministic incremental $(1+\eps)$-approximate data structures supporting:
  \begin{itemize}
    \item single-edge insertions in $O(n^{1.529}\cdot (1/\eps)\cdot \log(C/\eps))$ worst-case time and
      approximate shortest path-reporting queries in
      $O(n^{0.529}+|P|)$ time,
    \item single-edge insertions in $O(n^{1.407}\cdot (1/\eps)\cdot \log(C/\eps))$ worst-case time and
      approximate shortest path-reporting queries in
      $O(n^{1.407}+|P|)$ time,
  \end{itemize}
      where $P$ is the reported \emph{not necessarily simple} path.
\end{theorem}
Interestingly, both trade-offs in Theorem~\ref{t:incremental-appr-sum} match the best-known dynamic matrix inverse-based bounds for fully dynamic
\emph{transitive closure} in the regime of single-edge updates and single-pair queries~\cite{BrandNS19}.
For comparison, the state-of-the-art fully dynamic $(1+\eps)$-approximate all-pairs distances data structure~\cite{BrandN19}
has $O(n^{1.863}/\eps^2)$ update time and $O(n^{0.666}/\eps^2)$ query time.

Whereas in our approximate data structures the reported paths need not be simple,
in the important scenario with $\eps=O(1)$ and $C=O(1)$ (e.g., for unweighted graphs), the reported non-necessarily simple
path (that is nevertheless approximately shortest in terms of length, but not necessarily in terms of hop-length) may contain only a constant factor more edges than the shortest simple path (since the minimum allowed weight is $1$).
As a result, the latter data structure in Theorem~\ref{t:incremental-appr-sum} may also be used
to obtain an $O(n^{1.407})$ worst-case update/query time trade-off for simple-path-reporting incremental \emph{reachability} (by setting, e.g., $\eps=1$).

Finally, for incremental \emph{exact} shortest paths in \emph{unweighted} directed graphs, we obtain the following trade-offs.
\begin{theorem}\label{t:incremental-exact-sum}
  Let $G$ be an unweighted digraph.
  There exist deterministic incremental data structures supporting:
  \begin{itemize}
    \item single-edge insertions and shortest path-reporting queries in $O(n^{1.62})$ worst-case time,
    \item insertions of at most $n$ incoming edges of a single vertex, and single-source shortest paths tree-reporting
      queries in $O(n^{1.724})$ worst-case time.
  \end{itemize}
\end{theorem}
The former bound is polynomially smaller than the best-known $O(n^{1.724})$ bound on the update/query
time of a fully dynamic all-pairs distances data structure~\cite{BrandNS19,Sankowski05}.
The latter bound is even more interesting -- it is still not known whether there exists an exact fully dynamic
data structure maintaining single-source distances in an unweighted graph that would achieve
subquadratic update time, even for a fixed source $s$.
By our incremental bound and Lemma~\ref{l:offline}, we obtain that an \emph{offline} fully dynamic data structure with
subquadratic update time is possible.

\section{Preliminaries}
In this paper we deal with (possibly weighted) \emph{directed} graphs.
We write $V(G)$ and $E(G)$ to denote the sets of vertices and edges of $G$, respectively. We omit $G$ when the graph in consideration is clear from the context.
A graph $H$ is a \emph{subgraph} of $G$, which~we~denote by $H\subseteq G$, if
and only if $V(H)\subseteq V(G)$ and $E(H)\subseteq E(G)$.
We write $e=uv\in E(G)$ when referring to edges of $G$ and use $\wei_G(uv)$
to denote the weight of $uv$ (in case of weighted digraphs).
We call $u$ the tail of $e$, and $v$ the head of $e$.
If $uv\notin E$, we assume $\wei_G(uv)=\infty$.

For some set $E'\subseteq V\times V$, we denote by $G\cup E'$ the graph $(V(G),E(G)\cup E')$.
Similarly, we denote by $G\setminus E'$ the graph $(V(G),E(G)\setminus E')$.
We also use the notation $G-e:=G\setminus \{e\}$.

A sequence of edges $P=e_1\ldots e_k$, where $k\geq 1$ and $e_i=u_iv_i\in E(G)$, is called
an $s\to t$ path in~$G$ if $s=u_1$, $v_k=t$ and $v_{i-1}=u_i$ for each $i=2,\ldots,k$.
For brevity, we sometimes also express~$P$ as a sequence of $k+1$ vertices $u_1u_2\ldots u_kv_k$ or as a subgraph of $G$ with vertices $\{u_1,\ldots,u_k,v_k\}$
and edges $\{e_1,\ldots,e_k\}$.
The \emph{length} of a path $P$ equals $\sum_{i=1}^k\wei_G(e_i)$.
The \emph{hop-length} $|P|$ is equal to the number $k$ of edges in $P$.
We also say that $P$ is a \emph{$k$-hop path}.
For convenience, we sometimes consider a single edge $uv$ as a path of hop-length $1$, as well
as zero length sequence is used to denote an empty path.
If $P_1$ is a $u \to v$ path and $P_2$ is a $v \to w$ path, we denote by $P_1\cdot P_2$ (or simply $P_1P_2$) a path obtained by concatenating $P_1$ with $P_2$.

An \emph{out-tree} $T_{\text{out}}\subseteq G$ is a subgraph of $G$ for which there exists a vertex $s\in V(T_\text{out})$ (the root),
such that each vertex $v\in V(T_{\text{out}})\setminus\{s\}$ has precisely one incoming edge in $T_{\text{out}}$,
and $s$ has no incoming edges in $T_{\text{out}}$.
Symmetrically, an \emph{in-tree} $T_{\text{in}}\subseteq G$ is a subgraph of $G$ for which there exists a vertex $t\in V(T_\text{in})$ (the root),
such that each vertex $v\in V(T_{\text{in}})\setminus\{t\}$ has precisely one outgoing edge in $T_{\text{in}}$,
and $t$ has no outgoing edges in $T_{\text{in}}$.

A vertex $t$ is \emph{reachable} from $s$, if there is an $s\to t$ path in $G$.
A \emph{single-source reachability tree} $T_G(s)$ from $s$ is an out-tree in $G$ whose
root is $s$ and $V(T_G(s))$ equals the set of vertices reachable from $s$ in $G$.

Two vertices $u,v\in V$ are strongly connected if both $u$ is reachable from $v$ and $v$ is reachable from~$u$.
Strong connectivity is an equivalence relation.
We use the term \emph{strongly connected components} (SCCs) to refer to the set $\sccs$ of equivalence classes of the strong connectivity relation.

For some partition $\mathcal{V}=\{V_1,\ldots,V_k\}$ of $V$, we define $G/\mathcal{V}$ to
be a graph with vertices $\mathcal{V}$ obtained from $G$ by contracting each subset $V_i$ into a single vertex labeled $V_i$ ($G[V_i]$
does not necessarily need to be a connected subgraph).
For any $xy\in E(G)$ such that $x\in V_i$ and $y\in V_j$, we have a corresponding edge $V_iV_j$ in $G/\mathcal{V}$ if and only if
$V_i\neq V_j$. As a result, $G/\mathcal{V}$ can be a multigraph.

\section{Fully dynamic path reporting in DAGs}\label{s:reach-dags}

In this section we assume that the dynamic graph $G$ remains acyclic at all times
and give dynamic path reporting algorithms under this assumption.
We will first show a fully dynamic algorithm for maintaining the topological order.
This is a crucial element in our point-to-point path reporting algorithm.
Finally, we will use another idea to construct an algorithm for reporting single-source reachability trees.
These results will be lifted to general digraphs in Section~\ref{s:reach-general}.

We will repeatedly make use of the following dynamic transitive closure data structure of~\cite{Sankowski04} that
allows for subquadratic updates and sublinear queries.

\begin{theorem}\label{t:dyn-tr}\textup{\cite{Sankowski04}}
  Let $G$ be a digraph and let $\delta\in (0,1)$.
  There exists a data structure that
  supports single-edge insertions/deletions to $G$ in $O(n^{\omega(1,\delta,1)-\delta}+n^{1+\delta})$
  worst-case
  time, and point-to-point reachability queries in $G$ in $O(n^{\delta})$ worst-case time.
  The data structure is Monte Carlo randomized and produces answers correct with high probability.
\end{theorem}

It is important to note that as long as the data structure does not err (which
happens with low probability over its random choices for the fixed sequence of $\poly{(n)}$ updates), the produced answers depend only on the current graph (and not the previous answers)
and thus are unique. Therefore, the data structure can be obviously used against an adaptive adversary.

\subsection{Fully dynamic topological order}
Let us identify $V$ with $\{1,\ldots,n\}$.
Let $\pi:V\to \{1,\ldots,n\}$ be a permutation such that $uv\in E$ implies
$\pi(u)<\pi(v)$. Upon initialization, some $\pi$ can be computed in linear
time using one of the classical algorithms.

\newcommand{\dtc}{\mathcal{D}}

We will maintain a data structure $\dtc$ of Theorem~\ref{t:dyn-tr} on $G$ and pass
all the issued edge updates to it after updating the topological order.
Moreover, we will store the vertices $V$ in an array $A$ sorted according to $\pi$,
i.e., we have $A[i]=v$ iff $\pi(v)=i$.
Since $\pi$ is simply an inverse of $A$, we will only focus on maintaining $A$;
all the changes to $A$ can be reflected in $\pi$ in a straightforward way.

Now let us describe how to handle updates. If the update is a deletion of an edge then
we do not have to do anything, since $\pi$ remains a topological order
of $G-e$. So suppose we are inserting an edge $e=uv$.
Again, if we currently have $\pi(u)<\pi(v)$, the topological order does not need to be updated.
Now consider the case that $\pi(u)>\pi(v)$.
We need to modify the topological numbering only for vertices
$w$ currently satisfying $\pi(w)\in [\pi(v),\pi(u)]$, as for
the remaining vertices the topological order can remain the same.
Let us call the set of these vertices $W$.
In other words,~$W$ contains the vertices between $v$ and $u$ in $A$ (including $u$ and $v$).

Let $T\subseteq W$ be the vertices of $W$ reachable from $v$ before the insertion,
including $v$.
Let $S\subseteq W$ be the vertices of $W$ that can reach $u$ before the insertion,
including $u$.
Note that each of $S,T$ can be found by issuing $|W|=O(n)$ queries
to the data structure $\dtc$, i.e., in $O(n^{1+\delta})$ worst-case time.

We will update $A$ in the following way.
Let $Z=W\setminus (S\cup T)$.
The subarray $A[\pi(v)..\pi(u)]$
will be replaced by a sequence of vertices $S\cdot Z\cdot T$,
where each of the sets $S,T,Z$ is ordered according to~$\pi$ (before the insertion).
The correctness of such a change follows from the following lemma.

\begin{lemma}
  Let $S,T,Z$ be defined as above. Then, after the edge insertions:
  \begin{itemize}
  \item no vertex $t \in T$ can reach a vertex $s\in S$,
  \item no vertex $t \in T$ can reach a vertex $z \in Z$,
  \item no vertex $z \in Z$ can reach a vertex $s \in S$.
  \end{itemize}
\end{lemma}
\begin{proof}
Note that no vertex $t\in T$ can reach $s\in S$ as otherwise
there would be a path $v\to t\to s\to u$ in $G$, which, along
with the inserted edge $uv$, would form a cycle.
This would contradict that $G$ is acyclic.
In particular, we have $S\cap T=\emptyset$, $v\in T$, and $u\in S$.

Observe that no vertex $z\in Z$ can be reached from a vertex $t\in T$,
since, by definition, $T$ contains all vertices in $W$ reachable from $v$.
Similarly, no vertex in $Z$ can reach a vertex from~$S$.
\end{proof}

Finally, observe that if for two vertices $(x,y)$, one of the following holds: $(x,y)\in S^2$,
$(x,y)\in T^2$, $(x,y)\in Z^2$,
$\pi(x)<\pi(v)$, or $\pi(y)>\pi(u)$ before the insertion,
then the same holds also after the insertion.
This is because we don't change the relative order of vertices in
each of the sets $S$, $T$, $Z$, $\{x:\pi(x)<\pi(v)\}$, $\{y:\pi(y)>\pi(u)\}$.

Let $\rho$ be the smallest number such that $\omega(1,\rho,1)=1+2\rho$.
With the current best known bounds on the values $\omega(1,\cdot,1)$,
we have $\rho\approx 0.529$~\cite{GallU18}.
By setting $\delta=\rho$, we obtain the following lemma.
\begin{lemma}\label{l:topsort}
  Let $G$ be a digraph. There exists a data structure supporting
  fully dynamic single-edge updates to $G$ \emph{that keep $G$ acyclic},
  and maintaining a topological order $\pi$ of $G$ in $O(n^{1+\rho})=O(n^{1.529})$
  worst-case time. The algorithm is Monte Carlo randomized. With high probability,
  the maintained topological order is correct and is uniquely determined by
  the sequence of updates.
\end{lemma}
\begin{proof}
  Note that here randomization is only used inside the data structure $\dtc$.
  Consequently, high probability correctness follows by Theorem~\ref{t:dyn-tr}.
  As the answers produced by that component are uniquely determined by the graph updates (as the graph itself is),
  so is the maintained topological order.
\end{proof}

\subsection{Point-to-point path queries}\label{s:dag-point-to-point}
In order to support point-to-point path queries under fully dynamic edge updates, we will use two
data structures. The first one is a
data structure $\dtc$ of Theorem~\ref{t:dyn-tr} with $\delta$ set to $\rho$. The other data structure is that of Lemma~\ref{l:topsort},
maintaining a topological order $\pi$. No additional information is maintained, and so
in the updates we are simply passing each edge update to those data structures.

Now consider queries. Suppose we are requested to find some $s\to t$ path, where $s,t\in V$ are the query vertices.
Using a single query to $\dtc$ we can check whether such a path exists
in $O(n^\rho)$ time.
If this is not the case, we are done.
Otherwise, we infer that $\pi(s)<\pi(t)$.

The algorithm for constructing an $s\to t$ path (that is known to exist), is recursive.
If $s=t$, then clearly an empty path can be returned.

Otherwise, we scan through the outgoing edges $e_x=sx$ of $s$ in the order
of $\pi(x)$.
Note that for each such edge we have $\pi(x)>\pi(s)$ and the values $\pi(x)$ are distinct.
When scanning the edge $e_x$, we stop if $x$ can reach $t$ in $G$.
Each such test takes $O(n^{\rho})$ time using a single query to $\dtc$.
Then we recursively construct a path $P=x\to t$ and return the path $e_xP$.

Observe that since an $s\to t$ path exists, some edge $e_x$ will surely
lead to a recursive call: this will happen for $e_x$ with minimum $\pi(x)$
such that a path $x\to t$ exists in $G$.

\begin{lemma}\label{l:top-path}
  If $s$ can reach $t$ in $G$ then the above algorithm constructs an $s\to t$ path in\linebreak ${O((\pi(t)-\pi(s))\cdot n^\rho+1)}$
  time.
\end{lemma}
\begin{proof}
  We proceed by induction on $\pi(t)-\pi(s)$. If $\pi(t)-\pi(s)=0$, i.e., $s=t$, then the algorithm obviously finishes
  with a correct answer in $O(1)$ time.

  Suppose $\pi(s)<\pi(t)$. Consider a path $P=s\to t$ in $G$ whose first edge $e=sx$ has the minimum value
  $\pi(x)$.
  Then, the algorithm issues exactly one query to $\dtc$ for the existence of a path $y\to t$
  for each $y\in V$ with $\pi(y)\in (\pi(s),\pi(x))$. This amounts to at most $\pi(x)-\pi(s)-1$ queries.
  Each such query returns a negative answer.
  Then, the answer to a subsequent $x\to t$ query is affirmative.
  As a result, the algorithm recursively searches for an $x\to t$ path.
  Since $\pi(t)-\pi(x)<\pi(t)-\pi(s)$, by the inductive hypothesis,
  the $x\to t$ path will be constructed in $O((\pi(t)-\pi(x))\cdot n^\rho+1)$ time.
  So the total time needed to construct an $s\to t$ path
  is $O(1+((\pi(x)-\pi(s))+(\pi(t)-\pi(x))\cdot n^\rho+1)=O(1+(\pi(t)-\pi(s))\cdot n^\rho)$, as desired.
\end{proof}

Using the above, we obtain the following lemma.
\begin{lemma}\label{l:dag-path-query}
  Let $G$ be an acyclic digraph. There exists a data structure supporting
  fully dynamic single-edge updates to $G$ \emph{that keep $G$ acyclic},
  and point-to-point path queries, both in $O(n^{1+\rho})=O(n^{1.529})$
  worst-case time. The algorithm is Monte Carlo randomized and produces
  correct answers with high probability and against an adaptive adversary.
\end{lemma}
\begin{proof}
  By Lemma~\ref{l:topsort}, with high probability (if the data structure $\dtc$ does not err) the maintained topological
  order is uniquely determined by the updates.
  The query algorithm does not use randomization, so that the produced answer does not reveal any information about the random bits used
  by the internal data structures.

  By Theorem~\ref{t:dyn-tr} and Lemma~\ref{l:topsort}, each edge update is processed
  in $O(n^{1+\rho})$ worst-case time.
  Since we have $\pi(t)-\pi(s)\leq n$ for any query $(s,t)$, by Lemma~\ref{l:top-path}, an $s\to t$ path -- if one exists --
  is constructed in $O(n^{1+\rho})$ worst-case time.
\end{proof}

\subsection{Single-source reachability tree queries}\label{s:dag-tree}
Suppose we want to compute a single-source reachability tree from a query vertex $s\in V$.
We could, in principle, reuse the path reporting procedure developed for point-to-point
queries to construct such a tree. By taking for each $t\in V\setminus\{s\}$
that is reachable from $s$, the ultimate edge $e_t$ on the $s\to t$ path, we obtain
an out-tree rooted at $s$.
Such an edge $e_t=yt$ could be computed by running a single step of the recursive
algorithm on the reverse graph of $G$.
Unfortunately, in general, finding that edge could require $\pi(t)-\pi(y)$ reachability
queries to $\dtc$. This could result in $\Theta(n^{1+\rho})$ time per just a single edge of the tree
when $\pi(t)-\pi(y)=\Theta(n)$.

As a consequence, we need to use a different approach.
Let $S$ be the set of all vertices reachable from $s$ in an acyclic digraph $G$.
Recall that $S$ can be computed using $O(n)$ reachability queries issued to $\dtc$
in $O(n^{1+\rho})$ time. We rely on the following observation that holds for DAGs.

\begin{observation}\label{o:outtree}
  For each $t\in S\setminus\{s\}$, let $e_t=vt$ be an arbitrary edge of $G$ with $v\in S$.
  Then, the edges $\{e_t:t\in S\setminus\{s\}\}$ form an out-tree on $S$ rooted at $s$.
\end{observation}

By the above observation, it is enough to pick, for each $t \in S\setminus\{s\}$, \emph{any} incoming
edge with its tail in $S$.
In order to work against an adaptive adversary we will pick
these edges in a consistent deterministic manner.
For example, we could pick for each $t$ an incoming edge $e_t=vt$ where $v\in S$
has the minimum possible label (recall that we identify $V$ with $\{1,\ldots,n\}$).

To achieve that we use an approach reminiscent of the algorithms for computing
minimum witnesses for boolean matrix multiplication~\cite{CzumajKL07}.
Let $\Delta<n$ be an integer parameter to be chosen later.
We will maintain $q=\lceil n/\Delta \rceil$ data structures $\dtc_1,\ldots,\dtc_q$
of Theorem~\ref{t:dyn-tr}, where the underlying graph $G_i$ maintained in $\dtc_i$
is defined as follows.
Let $V',V''$ be two copies of $V$. Denote by $v'\in V'$ and $v''\in V''$ the corresponding copies of a vertex $v\in V$.
We have $V(G_i)=V\cup V'\cup V''$ and $G\subseteq G_i$.
For each edge $uv\in E(G)$, we have $uv\in E(G_i)$ and $u'v''\in E(G_i)$.
Moreover, for each $u\in V$ such that $u\in [(i-1)\cdot \Delta+1,i\cdot \Delta]$,
we add an edge $uu'$ to $G_i$.

\begin{lemma}\label{l:detect-interval}
  Let $u,v\in V(G)$, $u\neq v$. Then, a path $u\to v''$ exists in $G_i$ if and only if there exists
  a path $P=u\to v$ in $G$ such that the ultimate edge $yv$ of $P$ satisfies
  $y\in [(i-1)\cdot \Delta+1, i\cdot \Delta]$.
\end{lemma}
\begin{proof}
  For the forward direction, let $P=P'\cdot yv$. Then, by $G\subseteq G_i$ it follows that $P'$ is a $u\to y$ path in $G_i$.
  But since ${y\in [(i-1)\cdot \Delta+1, i\cdot \Delta]}$, there exist edges $yy'$ and $y'v''$ in $G_i$ as well.
  Therefore, $G_i$ indeed contains a $u\to v''$ path $P'\cdot yy'\cdot y'v''$.

  Now suppose there is an $u\to v''$ path $Q$ in $G_i$. Since the vertex $v''$ has only incoming edges from $V'$,
  and a vertex from $V'$ has an incoming edge (necessarily one, from $V$) only if that vertex lies in
  $[(i-1)\cdot \Delta+1,i\cdot \Delta]$,
  $Q$ is of the form $Q'\cdot yy'\cdot y'v''$, where $y\in V$ is such that $yv\in E(G)$ and ${y\in [(i-1)\cdot \Delta+1, i\cdot \Delta]}$.
  Moreover, $Q'\subseteq G_i[V]=G$ as $V'\cup V''$ has no outgoing edges to $V$.
  So, since $yv\in E(G)$, $Q'\cdot yv$ is an $u\to v$ path in $G$ and its penultimate
  vertex indeed lies in the desired interval.
\end{proof}

Every issued edge update to $G$ is passed to each of the $q$ data structures $\dtc_i$. Note that
such an update translates to two edge updates to $G_i$:
one in $V\times V$ and one in $V'\times V''$.
As a result, an edge update is processed in $O((n/\Delta)\cdot n^{1+\rho})$ worst-case time.

Now, at query time, in order to find for $t\in S\setminus\{s\}$ an incoming edge
$vt$ such that $v\in S$ and $v$ is minimum possible, we first find the smallest $j\in\{1,\ldots,q\}$
such that there is an $s\to t$ path with the penultimate vertex in the interval $[(j-1)\cdot \Delta+1,j\cdot \Delta]$.
By Lemma~\ref{l:detect-interval}, this can be achieved in $O((n/\Delta)\cdot n^\rho)$ time by issuing $q=O(n/\Delta)$ queries,
one per each of the data structures~$\dtc_i$.
Afterwards, in $O(\Delta)$ time we iterate over at most $\Delta$ edges of $t$ coming from
vertices in the interval $[(j-1)\cdot \Delta+1,j\cdot \Delta]$ in order to locate the desired minimum
labeled vertex $v$.
Since there are $O(n)$ different vertices $t$, finding the desired incoming
edges for all of them costs $O((n/\Delta)\cdot n^{1+\rho}+n\Delta)$ worst-case time.
By setting $\Delta=n^{(1+\rho)/2}$, we obtain the following.

\begin{lemma}\label{l:dag-tree-query}
  Let $G$ be an acyclic digraph. There exists a data structure supporting
  fully dynamic single-edge updates to $G$ \emph{that keep $G$ acyclic},
  and reporting a single-source reachability tree from any query vertex, both in $O(n^{(3+\rho)/2})=O(n^{1.765})$
  worst-case time. The algorithm is Monte Carlo randomized and produces
  answers correct with high probability against an adaptive adversary.
\end{lemma}
\begin{proof}
  By Theorem~\ref{t:dyn-tr}, each of the data structures $\dtc_1,\ldots,\dtc_q$ produces only correct answers with high probability.
  The answers produced by our algorithm are uniquely determined by the graph -- we always choose
  a minimum labeled feasible edge for Observation~\ref{o:outtree}. Hence, the answers
  of the algorithm do not reveal any random bits used by the underlying data structures.
\end{proof}

The Lemmas~\ref{l:dag-path-query}~and~\ref{l:dag-tree-query} combined yield Theorem~\ref{t:dags}.
\section{Overview of the remaining results}
\subsection{Fully dynamic path reporting in general digraphs}\label{s:path-overview}
The path finding algorithm behind Lemma~\ref{l:dag-path-query} fails for digraphs with non-trivial (i.e., consisting of at least two vertices)
strongly connected components. In order to deal with this issue, we apply the usual idea of solving
the problem separately on the ``acyclic'' part of the graph, and separately on the individual
strongly connected components. Of course, the acyclic part corresponds to the \emph{condensation} $G/\sccs$,
where $\sccs$ denotes the family of strongly connected components of $G$.

Fully dynamic path reporting inside strongly connected components
against an adaptive adversary is a challenge by itself, and no prior
tools for this task have been developed.
However, the known combinatorial methods~\cite{BernsteinGS20}
allow us to solve the ``strongly connected'' problem in subquadratic amortized time in the \emph{decremental} setting. This
often captures some of the critical difficulties of the fully dynamic setting.
To apply this tool, however, we need to abandon handling edge insertions
and deletions in a uniform way, as is typical in algebraic dynamic graph algorithms
(such as Theorem~\ref{t:dyn-tr})
based on path counting or dynamic matrix inverse.

In comparison to the algorithm for acyclic graphs from Section~\ref{s:dag-point-to-point},
in order to handle insertions,
the algorithm for general graphs operates in phases of $F$ edge insertions.
At the beginning of each phase, a path reporting decremental strongly connected components
data structure $\decscc$ is initialized for the current graph~$G$.
This data structure maintains a graph~$G^-$ defined as the graph at the beginning
of the current phase minus the edges deleted in that phase.
The data structure~$\decscc$ can be extended (in a standard way, see Lemma~\ref{l:decscc-extension})
to also maintain the condensation $G^-/\sccs$ efficiently,
where $\sccs$ denotes the set of strongly connected components of~$G^-$.
If there were no edge insertions issued, given a query $(s,t)$, applying the algorithm of Lemma~\ref{l:dag-path-query} to
the condensation $G^-/\sccs$ could produce, in $\Ot(n^{1+\rho})$ time,
a path $P$ in $G^-/\sccs$ between components $X,Y\in \sccs$ such that $s\in X$ and $t\in Y$.
Then, the data structure $\decscc$ and the condensation itself could be used
to lift the path $P$ to an actual $s\to t$ path in $G$, again in $\Ot(n^{1+\rho})$ time.

However, in presence of insertions,
to compute a desired $s\to t$ path upon query,
one first needs to identify which of the current phase's inserted edges (and in what order) necessarily
appear on a sought $s\to t$ path $P$.
This can be decided in $\Ot(n^{1+\rho}+nF)$ time using a simple but powerful generalization of Theorem~\ref{t:dyn-tr} -- given in Lemma~\ref{l:dyn-tr-general} --
that efficiently maintains an \emph{incremental} subset of rows/columns
of the transitive closure matrix.
Moreover, using Lemma~\ref{l:dyn-tr-general}, one can compute a partition of $G^-$ into clusters,
such that (at most $F+1$) individual maximal subpaths of $P$ entirely contained in $G^-$ can be sought
in separate clusters (see Lemma~\ref{l:query-partition}). This enables constructing them in $\Ot(n^{1+\rho})$ total time,
instead of $\Ot(|F|\cdot n^{1+\rho})$ time, which one would need to pay if each
of the subpaths was computed in the entire graph $G^-$.
For details, see Section~\ref{s:reach-general}.

\subsection{Fully dynamic reachability tree reporting in general digraphs}

Unfortunately, Observation~\ref{o:outtree} does not hold for general graphs:
given some source $s\in V$,
choosing an arbitrary incoming edge (e.g., that with the minimum label) from each vertex reachable from $s$
might lead to a disconnected (in the undirected sense) graph containing cycles.

In order to deal with this problem, we could, again, apply Observation~\ref{o:outtree} to the condensation
$G/\sccs$ (with the source set to the SCC containing $s$). The obtained
out-tree $T'$ in $G/\sccs$ could then be extended to a single-source reachability tree $T$ from $s$ in $G$ in two steps:
first, expand each vertex of $T'$ (i.e., a strongly connected
component $S$ reachable from $s$) into
a \emph{sparse} strongly connected subgraph of $S$ using the data structure $\decscc$.
Then, compute a single-source reachability tree from $s$ in the obtained subgraph of $G$
using any graph search procedure in $O(n)$ time.

For similar reasons as applied to path reporting, we need to operate in phases of edge insertions,
so that the strongly connected components of $G^-$ (defined as in Section~\ref{s:path-overview}) only split.
However, it is not clear how to efficiently handle the condensation $G^-/\sccs$ (which, critically, \emph{does not} include
some of the original edges of the decremental graph $G^-$)
using the algebraic data structure of Theorem~\ref{t:dyn-tr},
so that an interval of length $\Delta$ containing the minimum labeled tail of an incoming edge
can be located efficiently.
Recall that when the strongly connected components split, the condensation undergoes vertex splits, and each
vertex split (revealed online) might require $\Theta(n)$ edges changing endpoints.
Moreover, each of the $O(n^2)$ original edges of $G^-$ can be inserted, at some point,
to $G^-/\sccs$.
Consequently, we might need to perform $\Theta(n^2)$ edge updates on the data structure
of Theorem~\ref{t:dyn-tr} if we want it to reflect $G^-/\sccs$.
However, the updates in Theorem~\ref{t:dyn-tr} are relatively costly and we could not afford
performing $\Theta(n^2)$ such updates within a single phase.

We avoid the above problem by picking, for each vertex of $G^-$ (and not $G^-/\sccs$), the minimum-labeled incoming edge wrt. the topological order $\pi$
of $G^-/\sccs$, instead of a minimal edge wrt. an arbitrary order that is fixed initially.
This, of course, leads to other complications: the topological order of $G^-/\sccs$
evolves in time.
However, we can maintain a topological \emph{labeling} $\pi$ of the dynamic graph $G^-/\sccs$ that admits
a certain \emph{nesting property} (see Lemma~\ref{l:decscc-extension}). This property guarantees
that if the interval $[\pi(S),\pi(S)+|S|-1]$ for $S\in \sccs$ is contained in some interval $I_j$
of the form $[(j-1)\cdot \Delta+1,j\cdot\Delta]$, then we will have $\pi(S')\subseteq I$
for any $S'\subseteq S$ that becomes an SCC of $G^-$ in the future.
We call such an SCC $S$ \emph{non-special}, and all other SCCs \emph{special}.
In other words, the topological order of all vertices within a non-special component $S$ is fixed, up to the interval~$I_j$ it
is currently contained in.
This enables us to locate the minimum labeled edges coming
from non-special SCCs using the approach of Section~\ref{s:dag-tree}.

Luckily, the number of \emph{special} SCCs that do not fall into the above category
is always $O(n/\Delta)$, so we can handle them using a different, more straightforward approach.

The fully dynamic tree-reporting data structure constitutes the most technically involved part
of this paper and heavily relies on the developments of Section~\ref{s:reach-general}. Details can be found
in Section~\ref{s:tree-general}.

\subsection{Deterministic incremental algorithms with subquadratic worst-case bounds}
In order to obtain Theorems~\ref{t:incremental-reach-sum},~\ref{t:incremental-appr-sum},~and~\ref{t:incremental-exact-sum},
we follow the general approach behind the subquadratic trade-offs for
reachability -- based on either path counting~\cite{DemetrescuI05} or dynamic matrix inverse~\cite{BrandNS19, Sankowski04}.
Namely, the algorithms operate in phases of $F=O(n^\alpha)$ updates.
Each phase starts with a computationally heavy recomputation step
based on rectangular matrix multiplication, that, roughly speaking,
recomputes the reachability/shortest paths matrix $A$ for all pairs of vertices $u,v\in V$
in the graph $G_0$ (equal to $G$ when the phase starts) based on the analogously defined matrix $A'$ and the $O(n^\alpha)$ updates from the previous phase.
The cost of such a recomputation step should be thought of being amortized over the
$F$ updates of the phase, and this is the only source of amortization.
However, there is a well-known standard technique (also used in the previous dynamic algebraic graph algorithms, e.g.~\cite{BrandNS19, DemetrescuI05, Sankowski04}) for converting such amortized bounds into worst-case bounds by maintaining two copies of the data structure
that switch their roles every $F/2$ updates.

The queries are answered using the matrix $A$ and possibly some auxiliary data structure
build on the current phase's updates that depends on the trade-off we want to achieve.

In the fully dynamic algorithms for reachability, the matrix $A$ stores path counts modulo a prime \cite{DemetrescuI05},
or values of certain multivariate polynomials at suitable points modulo a prime~\cite{BrandNS19, Sankowski04}.
In the fully dynamic shortest paths algorithms~\cite{BrandN19, BrandNS19, Sankowski05},
the matrix $A$ stores coefficients of low-degree terms of multivariate polynomials with suitably chosen coefficients.
All these fully dynamic algorithms critically rely on the Zippel-Schwartz lemma and thus
are inherently Monte-Carlo randomized.

On the other hand, we show that in the incremental setting a much simpler idea is sufficient.
Roughly speaking, in our incremental algorithms, the matrix $A$ stores \emph{paths} in $G_0$.
These paths are implemented as optimal purely functional concatenable deques of~\cite{KaplanT99} to allow for
convenient and efficient path manipulations, such as concatenation and iteration.
More specifically, for reachability these are simple paths, for approximate shortest paths $A$ stores
approximate shortest paths, whereas for exact shortest paths, $A$ stores short paths of length no more than a certain threshold $h$.

For reachability, the recomputation is performed using boolean matrix multiplication,
for approximate shortest paths, we use the $(1+\eps)$-approximate min-plus product~\cite{Zwick02}, whereas
for exact shortest paths -- the bounded exact min-plus product~\cite{Zwick02}.
Turning these respective (rectangular) matrix products of number matrices into products of path matrices
is possible since each of these products allows for computing \emph{witnesses} deterministically~\cite{AlonGMN92, Zwick02}.

The trade-offs that we obtain for incremental reachability (Section~\ref{s:incremental-reach})
match the best known fully dynamic transitive closure trade-offs
for single-edge and single-vertex-incoming-edges updates obtained via dynamic matrix inverse~\cite{BrandNS19}. However,
these fully dynamic trade-offs seem hopelessly randomized and do not allow for path reporting.

Interestingly, for $(1+\eps)$-approximate shortest paths in the single-edge updates setting (Section~\ref{s:incremental-appr}), we obtain the same trade-offs (up to polyloagrithmic factors) as in the case of reachability.
To this end, we need to carefully control the error and perform recomputations
using a hierarchical binary-tree-like circuit, as opposed to a linear path-like circuit that is sufficient for reachability.

Finally, for incremental exact shortest paths in unweighted graphs (Section~\ref{s:incremental-exact}),
we exploit a very simple observation (also used in e.g.~\cite{KarczmarzL19}), that a deterministically computed $\Ot(F)$-size hitting set of length-$n/F$ paths
remains valid and is $\Ot(F)$-sized throughout a phase of insertions if we simply augment it with the endpoints
of inserted edges.

\section{Fully dynamic strongly connected components}\label{s:scc}
As a warm-up, let us first consider fully dynamic maintenance of strongly connected components.
We start with the following variant of the data structure from Theorem~\ref{t:dyn-tr}
which will be crucial to all our developments in general graphs.
\begin{lemma}\label{l:dyn-tr-general}
  Let $G=(V,E)$ be a digraph. Let $R\subseteq V$.
  There exists a data structure explicitly maintaining the information whether there exists
  an $r\to v$ path in $G$ for each of the pairs $(r,v)\in R\times V$ and supporting each of the following update operations
  in $O(n^{1+\rho}+n\cdot |R|)$ worst-case time:
  \begin{itemize}
    \item single-edge insertions/deletions to $G$,
    \item adding a new vertex to the set $R$,
    \item resetting $R$ to $\emptyset$.
  \end{itemize}
  The data structure is Monte Carlo randomized and produces correct answers with high probability.
\end{lemma}

\begin{proof}
  We follow the idea of \cite{Sankowski05}, which is an extension
  of the data structure of Theorem~\ref{t:dyn-tr} supporting single
  edge insertions and deletions in $O(n^{1 + \rho})$ worst-case time, and queries in $O(n^\rho)$ time (for $\delta:=\rho$).
  From \cite[Theorem~7]{Sankowski04}, maintaining the reachability
  information reduces, within the same update time and with high probability correctness, to maintaining the
  inverse of a matrix $A$ which is a symbolic adjacency matrix with variables replaced with random elements of the field $\mathbb{Z}/p\mathbb{Z}$ for a sufficiently
  large prime number $p=\Theta(\poly{n})$. Indeed, to maintain the
  reachability information from $R$ to $V$, it is enough to explicitly maintain
  the submatrix $(A^{-1})_{R,V}$ under the sequence of updates.

  The key idea is to write, after a single entry change, the updated matrix
  as the product of two
  matrices $A' = A\cdot B$. Here $B$ has a special form, i.e., it has
  non-zero elements only on the diagonal and one column (say $j$)
  where the update occurs~\cite[Theorem 4]{Sankowski05}. This implies that $B^{-1}$ has similar non-zero structure~\cite[Fact 5.4]{BergamaschiHGWW21full}.
  This way in the multiplication $A'^{-1} = B^{-1}A^{-1}
  = A^{-1} + (B^{-1} - I)A^{-1}$ only the $j$-th row of $A$ is used.
  In particular, we can obtain the submatrix $(A'^{-1})_{R,V}$ by
  the following equation
  $$(A'^{-1})_{R,V} = (A^{-1})_{R,V} + (B^{-1} - I)_{R,\{ j\} } \cdot (A^{-1})_{\{ j\},V}.$$

  The entries of $(A^{-1})_{\{ j\},V}$ can be found in $O(n^{1 + \rho})$ time using Theorem~\ref{t:dyn-tr}.
  Subsequently,
  recomputation of the part of the inverse can be done via a vector-vector
  product in $O(n\cdot |R|)$ time.
  To add a vertex $v$ to the set $R$ we additionally need to query the $v$-th row of
  $A^{-1}$ which again takes $O(n^{1 + \rho})$ time by Theorem~\ref{t:dyn-tr}.
\end{proof}

We will also need a near-optimal decremental strongly connected components data structure.
\begin{theorem}\label{t:decscc} \textup{\cite{BernsteinPW19}} Let $G$ be a directed graph. There exists a Las Vegas randomized data structure maintaining
  the strongly connected components $\sccs$ of $G$ \emph{explicitly}\footnote{That is, the algorithm maintains
  a mapping $s:V\to\sccs$ such that $s(v)$ is the SCC of $v$,
  and an identifier of each $s(v)$ is explicitly stored in a memory cell, so that one is notified
  every time that identifier changes.} subject to edge deletions in $\Ot(n+m)$ expected total time.
\end{theorem}
\begin{remark}\label{r:decscc}
  The data structure of Theorem~\ref{t:decscc} can be converted into a Monte Carlo data structure
  that runs in worst-case $\Ot(n+m)$ total time and is correct with high probability.
  Indeed, it is enough to maintain $\Theta(\log{n})$ independent copies of the data structure, such that each of them is terminated
  prematurely when
  its actual time used exceeds the expected running time by more than a fixed constant factor.
  A failure is declared if all of the copies are terminated prematurely. By Markov's inequality, this happens
  with low probability.
\end{remark}

Our algorithm operates in phases. Each phase spans $F$ edge \emph{insertions}.
At the beginning of a phase, we reinitialize a decremental data structure $\decscc$ of Theorem~\ref{t:decscc}
that will be responsible for maintaining
the components that are not affected by the insertions.
The remaining components will be handled using a data structure $\dtc$ of Lemma~\ref{l:dyn-tr-general}
with the (growing) set $R$ storing the (at most~$F$) heads of edges inserted in the current phase.

More specifically, using the data structures $\decscc$ and $\dtc$,
we will build a not very dense graph $H$ which will preserve
the strongly connected components of $G$. The following lemma defines such a graph and proves that it possesses the desired property.

\begin{lemma}\label{l:scc-cert}
  Let $G$ be a directed graph. Let $E^+$ be some set of at most $f$ edges.
  Let a graph $H$ on~$V$ with $O(nf)$ edges contain the following edges:
  \begin{enumerate}
    \item for each strongly connected component $S$ of $G$, a directed cycle on the vertices $S$,
    \item for each $uv\in E^+$, and each $w\in V$, an edge $vw$ if $v$ can reach $w$ in $G\cup E^+$,
      and an edge $wv$ if $w$ can reach $v$ in $G\cup E^+$.
  \end{enumerate}
  Then, for all $u,v\in V$, $u$ and $v$ are strongly connected in $G\cup E^+$ if and only if they
  are strongly connected in $H$.
\end{lemma}
\begin{proof}
  First, each edge $xy$ in $H$
  certifies the existence of an $x\to y$ path in $G\cup E^+$.
  As a result, if $u$ and $v$ are strongly connected in $H$, then
  they are strongly connected in $G\cup E^+$ as well.
  This proves the ``$\impliedby$'' direction.

  Now, suppose $u$ and $v$ are strongly connected in $G\cup E^+$.
  If these vertices are strongly connected in $G$ as well, they lie on a single cycle
  in $H$, so they are strongly connected in $H$ as well.

  So consider the case when $u$ and $v$ are not strongly connected in $G$.
  Then, there exists either a $u\to v$, or a $v\to u$ path through some edge $xy=e\in E^+$ --
  if there existed paths in both directions not using any edge from $E^+$, $u$ and $v$ would
  be already strongly connected in $G$.
  It follows that~$y$ is strongly connected with $u$ and $v$ in $G\cup E^+$.
  Consequently, we have edges $yu$, $yv$, $uy$, and $vy$ in~$H$.
  It follows that there exist paths $u\to v$ and $v\to u$ in $H$, so $u$ and $v$ are indeed strongly connected in $H$.
\end{proof}

Let $E^+$ denote the set of edges inserted in the current phase (and not deleted
by one of the updates in that phase). Let $G^-$ denote the graph
$G$ at the beginning of the phase minus the edges deleted in the current phase.
So we have $G=G^-\cup E^+$ at all times.
We maintain the invariant that $\dtc$ stores the graph $G$ (i.e., $\dtc$ is passed
all the fully dynamic updates), and the set~$R$ in the data structure $\dtc$ equals the set of heads
$\{y:xy\in E^+\}$. By Lemma~\ref{l:dyn-tr-general}, this can be guaranteed in $O(n^{1+\rho}+nF)$ worst-case time
per update. The set $R$ is reset to $\emptyset$ when a new phase starts.

The data structure $\decscc$, on the other hand, is only passed the edge deletions and
thus maintains the graph $G^-$.

Now, we leverage Lemma~\ref{l:scc-cert} as follows.
We use $\decscc$ to construct the first type of edges of the graph $H$ of
Lemma~\ref{l:scc-cert} (applied to $G:=G^-$) in $O(n)$ time -- recall that $\decscc$ stores
the SCCs of $G^-$ explicitly.
Moreover, the edges of $H$ of the second type are constructed by reading
the $O(nF)$-size information stored by $\dtc$.
Finally, to compute the strongly connected components of the current $G$
(which equals $G^-\cup E^+$), we run any classical linear-time strongly connected components algorithm
on the graph $H$. Recall that the graph $H$ has size $O(nF)$.

\begin{lemma}
  The strongly connected components of a digraph can be maintained explicitly
  subject to edge insertions and deletions in $\Ot(n^{1+\rho})=O(n^{1.529})$ amortized
  time per update. The algorithm is Monte Carlo randomized and correct
  with high probability.
\end{lemma}
\begin{proof}
  Since the data structure $\decscc$ of Theorem~\ref{t:decscc} is rebuilt in each $F$ insertions,
  the amortized cost of maintaining it is no more than $\Ot(n^2/F)$.
  By Lemma~\ref{l:dyn-tr-general}, the data structure $\dtc$ is updated in $O(n^{1+\rho}+nF)$ worst-case time
  per edge update.
  Both data structures are Monte Carlo randomized and correct with high probability (see Remark~\ref{r:decscc}).
  To optimize the running time, set $F=n^{1/2}$.
  Since $\rho\geq 1/2$, the $n^{1+\rho}$ term dominates the amortized update time.
\end{proof}
Finally, we note that as strongly connected components are defined in a unique way, the above algorithm can be used
against an adaptive adversary, as the uniqueness guarantees that the structure of $\sccs$ does not leak randomness.

\section{Fully dynamic path reporting in general digraphs}\label{s:reach-general}
In this section, we will prove the following theorem.
{\renewcommand\footnote[1]{}\tpathreporting*}

We start with the following black-box extension of decremental strongly connected components
maintenance. Similar extensions have been previously used in, e.g.,~\cite{BernsteinGS20, BernsteinGW20},
albeit with the goal of maintaining approximate single-source shortest paths.

\begin{lemma}\label{l:decscc-extension}
  Let $G$ be a digraph and suppose the strongly connected components $\sccs$ of~$G$
  are maintained explicitly as a mapping $s:V\to \sccs$ subject to edge deletions.

  Then, in additional $O(n+m\log{n})$ total time, we can also achieve the following.
  \begin{enumerate}
    \item Explicitly maintain \emph{topological labels} $\pi:\sccs\to [1,n]$ of the SCCs so that:
      \begin{itemize}
        \item All the intervals $[\pi(S),\pi(S)+|S|-1]$, $S\in \sccs$ are disjoint and form a partition of $[1,n]$.
        \item If for some $S'\in\sccs$, $S\neq S'$ there is a path from $S$ to $S'$ in $G$, then $\pi(S)<\pi(S')$.
        \item If an SCC $S$ splits due to an edge deletion
          into some number of smaller SCCs $S_1,\ldots,S_k$, then the intervals $[\pi(S_i),\pi(S_i)+|S_i|-1]$
          form a partition of $[\pi(S),\pi(S)+|S|-1]$.
      \end{itemize}
    \item Explicitly maintain a \emph{condensation} $G/\sccs$ of $G$, which is a \emph{simple} acyclic digraph $(\mathcal{S},E')$
      such that for each $X,Y\in \sccs$ there is an edge $XY\in E'$ if and only if
          there is an edge $e_{X,Y}=xy\in E$ such that $x\in X$ and $y\in Y$.
          Moreover, one such edge $e_{X,Y}\in E$ is maintained as well.
  \end{enumerate}
  The algorithm is deterministic and additionally, the maintained information depends only on the
  initial graph $G$ and the sequence of updates.
\end{lemma}
\begin{proof}
  Internally we maintain a condensation $G/\sccs$ that is not simple, i.e., $G/\sccs$ contains a separate edge
  $XY\in E(G/\sccs)$ for each $xy\in E(G)$ such that $x\in X$ and $y\in Y$, and $X\neq Y$ (i.e., $G/\sccs$ contains
  no self-loops).
  The condensation presented to the user, however, identifies all the parallel edges of this kind
  by providing any representative edge $e_{X,Y}$ from each set of parallel edges.

  For each $S\in\sccs$ we maintain a number $h(S)$ equal to the number of vertices in~$S$ in~$G$.
  These numbers can clearly be initialized in $O(n+m)$ time when the process starts.
  Consider some edge deletion of $e$ in~$G$. Let $\sccs'$ denote the strongly connected components
  in $G-e$.
  Suppose some $S\in \sccs$ splits into $S_1,\ldots,S_k\in \sccs'$ as a result of deleting $e$.
  Then, for all~$S_i$, except of maybe one, say wlog. $S_1$, the value $s(v)$ changes for all $v\in S_i$,
  i.e., the algorithm maintaining the mapping~$s$ spends at least $T\geq \sum_{i=2}^k |S_i|$ time on updating~$s$.
  We update the counters $h(S_i)$ for $i>1$ by iterating through all the vertices in these components.
  We then set $h(S_1)=h(S)-\sum_{i=2}^k h(S_i)$. This takes $O(T)$ time, so it can be charged
  to the cost of maintaining the mapping $s$.
  Let us now relabel $S_1,\ldots, S_k$ so that $h(S_1)\geq h(S_2),\ldots,h(S_k)$,
  i.e., we have $h(S_i)\leq h(S)/2$ for all $i=2,\ldots,k$.

  We now show how to update the (non-simple) condensation $G/\sccs$. The vertex $S$ is split into vertices $S_1,\ldots,S_k$
  in $(G-e)/\sccs'$.
  We initially make all the edges $SY$, where $Y\neq S$, previously incident to $S$ in $G/\sccs$,
  incident to $S_1$ in $(G-e)/\sccs'$, by simply relabeling the vertex $S$ to $S_1$.
  Next, we iterate through all edges $xy=e'\in E(G-e)$ such that $x\in \bigcup_{i=2}^k S_i$ or $y\in \bigcup_{i=2}^k S_i$.
  There are a few cases to consider. Suppose $x\in S_k$. If $y\notin S$ and $y\in Y\in \sccs\cap \sccs'$, then there was already
  an edge $SY$ corresponding to $xy$ in $G/\sccs$.
  We disconnect that edge from $S_1$ and add it to $(G-e)/\sccs'$ as an edge $S_kY$.
  If, on the other hand, $y\in S$, then there was no edge corresponding to $xy$ in $G/\sccs$.
  If additionally $y\in S_l$, where $k\neq l$, then we add an edge $S_kS_l$ corresponding
  to $xy$ to the graph $(G-e)/\sccs'$.
  We proceed similarly with the edges $xy$ with $y\in \bigcup_{i=2}^k S_i$.
  It is easy to verify that such a procedure updates the non-simple condensation correctly
  in $O\left(\sum_{i=2}^k\sum_{v\in S_i}\deg(v)\right)$ time.

  To update the topological labels, we simply compute the topological order
  of the induced subgraph $((G-e)/\sccs')[\{S_1,\ldots,S_k\}]$
  using any classical deterministic linear time algorithm.
  The induced subgraph has size (and can be constructed in) $O\left(\sum_{i=2}^k\sum_{v\in S_i}\deg(v)\right)$ time
  since each of its edges has at least one endpoint in $S_2,\ldots,S_k$.
  Finally, after relabeling $S_1,\ldots,S_k$ so that they are sorted
  topologically within $((G-e)/\sccs')[\{S_1,\ldots,S_k\}]$,
  we set $\pi(S_1):=\pi(S)$, and $\pi(S_i):=\pi(S_{i-1})+h(S_{i-1})$ for all $i=2,\ldots,k$.

  Finally, to bound the total update time spent on maintaining the condensation and the topological labels,
  note that whenever we pay $\deg(v)$ time for some
  work involving a vertex $v\in V$ (of~$G$) when processing a deletion of $e$, the size of a component that contains $v$ decreases by a factor of at least $2$
  after removing $e$.
  Hence, in the total update time, for each $v\in V$ we pay $\deg(v)$ at most $O(\log{n})$ times.
  Through all $v\in V$, the total update time is $O(m\log{n})$.
\end{proof}

The algorithm will again operate in phases spanning $F$ insertions forming a set $E^+$, where $F$ is to be chosen
later. Each phase will involve initializing a decremental data structure $\decscc$
as stated in Theorem~\ref{t:path-reporting}, along with the auxiliary data structures
from Lemma~\ref{l:decscc-extension}.
So, $\decscc$ will actually maintain a graph $G^-$ defined
as the graph $G_0$ from the beginning of the phase minus the edges of~$G_0$
deleted in the current phase.
Denote by $\sccs$ the strongly connected components of $G^-$.

Again, we will use a data structure $\dtc$ of Lemma~\ref{l:dyn-tr-general} with the set
$R$ storing the heads of the edges of $E^+$, but in
a slightly different manner.
Namely, when the phase proceeds, $\dtc$ will only accept edge deletions.
As a result, throughout the phase, $\dtc$ will also store the graph $G^-$.
However, when the phase ends, all the edges $E^+$ will be added to~$\dtc$ at once.

Additionally, we use a data structure~$\mathcal{Q}$ of Theorem~\ref{t:dyn-tr} (with $\delta=\rho$) in a similar
way as $\dtc$ in order to enable $O(n^\rho)$-time reachability queries on $G^-$ throughout the phase.

Given the source $s$ and target $t$, the query algorithm will first identify a
minimal subset $E^+_{s,t}$ of edges of $E^+$ such that
a path $P=s\to t$ exists in $G^-\cup E^+_{s,t}$,
along with the order $e_1,\ldots,e_k$ in which the edges $E^+_{s,t}$ appear on $P$.
Then, $P=P_0e_1P_1e_2P_2\ldots e_kP_k$ and each of the paths $P_0,\ldots,P_k$
is contained in $G^-$.
We will show that using $\decscc$ and $\dtc$, \emph{all} of these paths can
be found using $O(n)$ queries to $\mathcal{Q}$ and $O(n)$ path reporting queries
to $\decscc$ that will report paths of total length $O(n)$.

More specifically, we will first partition the graph $G^-$ so that
each of the paths $P_0,\ldots,P_k$ can be searched for in a separate
disjoint region of the graph $G$.
The following lemma describes such a partition in an even more general setting
that will prove useful when reporting a reachability tree instead of a path.

\begin{lemma}\label{l:query-partition}
  Let $E^+=\{u_1v_1,\ldots,u_kv_k\}$, where $k\leq F$. Let $s\in V$ and put $v_0:=s$.
  Let $V_s$ denote the set of vertices reachable from $s$ in $G^-\cup E^+$.

  In $O(n^{1+\rho}+nF)$ time one can partition $V_s$ into possibly empty and pairwise disjoint
  subsets $V_{s,0},\ldots,V_{s,k}$ such that for any $i$ and $z\in V_{s,i}$,
  \emph{every} $v_i\to z$ path in $G^-$ is fully contained in $G^-[V_{s,i}]$.

  Moreover, let $J=\{j:V_{s,j}\neq\emptyset\}$. For every $i\in J\setminus\{0\}$, there is a \emph{parent} $p(i)\in J\setminus\{i\}$
  such that $u_i\in V_{s,p(i)}$ and
  the edges $\{p(j)j:j\in J\setminus\{0\}\}$ form an out-tree on $J$ with root at $0$.
\end{lemma}
\begin{proof}
  We build the sets using the following procedure which gradually extends
  the set $V_s$ of vertices reachable from $s$ in $G^-\cup E^+$
  and inserts new elements to $J$.

  We start with $J=\{0\}$. We set $V_{s,0}$ to be the vertices
  that $s$ can reach in $G^-$. Note that those can be computed using $O(n)$
  queries to $\mathcal{Q}$.
  We initialize $V_s$ to $V_{s,0}$.

  Next, while for some $u_iv_i\in E^+$ we have $u_i\in V_s$ and $v_i\notin V_s$,
  we set $V_{s,i}$ to be the vertices reachable from $v_i$ in $G^-$
  minus those already in $V_s$.
  Moreover, we add $i$ to the set $J$.
  Note that since $v_i\in R$ in the data structure~$\dtc$, $V_{s,i}$ can
  be computed in $O(n)$ time by simply reading the vertices reachable from $v_i$
  from $\dtc$.
  Finally, if $u_i\in V_{s,j}$, then we set the parent $p(i)$ of $i$ to $j$
  and add $V_{s,i}$ to $V_s$.
  Observe that as we have $v_i\in V_s$ afterwards, this step is performed
  at most $|E^+|\leq F$ times.
  As a result, all the steps take $O(nF)$ time in total.

  Since the parent $p(i)$ of each $i$ added to $J$ is set to an element of $J$
  that was added to $J$ before~$i$, the edges $p(i)i$ indeed form an out-tree $T$
  rooted at the first element added to $J$, i.e., $0$, as desired.

  By the construction, it is clear that the sets $V_{s,i}$, $i\in J$, are pairwise
  disjoint. Let us now prove that their union $V_s$ indeed contains precisely
  the vertices reachable from $s$ in $G^-\cup E^+$.
  To prove that each vertex ever put in some $V_{s,i}$ is indeed reachable from $s$
  one can proceed by induction on the depth of $i$ in the tree $T$.
  For $i=0$ this is clear since $V_{s,0}$ contains precisely the vertices
  reachable from $s$ in $G^-\subseteq G^-\cup E^+$.
  If $i\neq 0$, then by the inductive hypothesis, every vertex in $V_{s,p(i)}$,
  in particular $u_{i}$, is reachable from $s$ in $G\cup E^+$.
  So $v_i$ is reachable as well since $u_iv_i\in E^+$.
  But every vertex in $V_{s,i}$ is reachable from $v_i$ in $G^-$,
  so it is also reachable from $s$ in $G^-\cup E^+$.

  Now suppose that $s$ can reach some $t\in V$ in $G^-\cup E^+$.
  Let $h$ be the minimum possible number of edges from $E^+$ on an
  $s\to t$ path in $G^-\cup E^+$.
  We prove that $t\in V_s$ by induction on $h$.
  If $h=0$, this means that $t$ is reachable from $s$ in $G^-$,
  and as a result we have $t\in V_{s,0}$.
  If $h\geq 1$, then let $P$ be some $s\to t$ path in $G^-\cup E^+$
  with exactly $h$ edges from $E^+$ and suppose $e=u_iv_i$ is the edge
  of $E^+$ that appears on $P$ last.
  Let us first observe that $v_i$ will be put, at some point, to
  some set $V_{s,j}$, possibly with $i=j$.
  Indeed, there exists an $s\to u_i$ path in $G^-\cup E^+$
  with less than~$h$ edges in~$E^+$.
  So, by the inductive hypothesis, we have $u_i\in V_s$.
  As a result, the main loop of the algorithm constructing the sets $V_{s,\cdot}$ will ensure that $v_i\in V_{s,j}$
  for some $j$ as well.
  But there is a path $v_i\to t$ in $G^-$, so whenever $v_i$ is included in some $V_{s,j}$, the construction ensures that all
  vertices reachable from $v_i$ (in particular $t$) in $G^-$ are included in $V_s$ as well.

  It remains to prove that if $z\in V_{s,i}$, then every $v_i\to z$ path in $G^-$
  contains vertices of $V_{s,i}$ exclusively.
  Suppose this is not the case and there exists a path $v_i\to w\to z$, where
  $w\in V_{s,j}$ for $j\neq i$.
  Note that $j$ could not be added to $J$ later than $i$ because when $V_{s,i}$
  is built, all vertices reachable from $v_i$ (in particular $w$) that were
  not in $V_s$ before are put into $V_{s,i}$.
  But if $j$ was added to $J$ earlier than $i$, then $w$ is reachable
  from $v_j$ in $G^-$ and so is $z$. So $z$ should be put into $V_s$
  at the time of insertion of $j$ into $J$ (or earlier), a contradiction.
\end{proof}

We will also need the following generalization of the path-finding
algorithm of Section~\ref{s:dag-point-to-point} whose performance was analyzed in Lemma~\ref{l:top-path}.
\begin{lemma}\label{l:path-in-condensation}
  Let $s,t\in V$ and let $W\subseteq V$. Suppose every $s\to t$ path in $G^-$
  lies within the subgraph $G^-[W]$.
  Let $\sccs_W=\{S\in \sccs: S\subseteq W\}$, where $\sccs$ are the strongly connected components of $G^-$.
  Let $S_s,S_t$ be the SCCs of $G^-$ containing $s$ and $t$ respectively.
  Given the relative topological order of $\sccs_W$ within the condensation $G^-/\sccs$,
  one can find an $S_s\to S_t$ path in $G^-/\sccs$ in $O(|W|\cdot n^\rho)$ time.
\end{lemma}
\begin{proof}
  Recall the recursive algorithm of Section~\ref{s:dag-point-to-point} that worked for acyclic graphs.
  Given a source~$s$
  and a target $t$, the algorithm searched for the topologically earliest vertex $x$ such that (1) there is an edge $sx$,
  and (2) there exists an $x\to t$ path.

  We apply the same algorithm to the acyclic graph $G^-[W]/\sccs_W$.
  We first find the components $G^-[W]/\sccs_W$ in $O(|W|)$ time using
  the explicitly stored mapping from vertices to SCCs,
  and their relative topological order using the labels $\pi$ from Lemma~\ref{l:decscc-extension}.
  Since the condensation $G^-/\sccs$ is explicitly maintained (Lemma~\ref{l:decscc-extension}) along with $\sccs$,
  for any $X\in \sccs$, we can query for the existence of an edge $SX$ in $G^-[W]/\sccs_W$ (as required by~(1)) in constant time\footnote{Technically speaking, such an efficient access requires the edges of the condensation maintained by Lemma~\ref{l:decscc-extension} also be stored in a hash table. Doing this introduces no additional asymptotic overhead.}.
  Note that a path between two vertices $X,Y$ of $G^-/\sccs$ exists
  if and only if a path between arbitrary $x\in X$, $y\in Y$
  exists in $G^-$. As a result, we can handle (2) using
  a single query to the data structure $\mathcal{Q}$.

  Recall that the algorithm of Section~\ref{s:dag-point-to-point} worked in time proportional to
  the number of vertices in the graph times $O(n^{\rho})$.
  So in our case, the algorithm takes $O(|\sccs_W|\cdot n^{\rho})=O(|W|\cdot n^{\rho})$ time to complete.
\end{proof}

Let us now describe the algorithm reporting an $s\to t$ path more formally.
In the first step, we compute a partition $V_{s,0},\ldots,V_{s,k}$  of Lemma~\ref{l:query-partition},
along with the set~$J$ and the tree structure on it
in $O(n^{1+\rho}+nF)$ time.
If $t\notin \bigcup_{i=0}^k V_{s,i}$, $t$ is not reachable from $s$ in $G$, so we are done.
Otherwise, suppose $t\in V_{s,j}$, and let $0=y_0,\ldots,y_\ell=j$ be the ancestors
of $j$ (including $j$) in the tree $J$, furthest to nearest.
Let the edges $u_iv_i\in E^+$, $i=1,\ldots,k$, be defined as in Lemma~\ref{l:query-partition},
By Lemma~\ref{l:query-partition}, there exists an $s\to t$ path
$P=P_0e_{y_1}\cdot P_1e_{y_2}\cdot \ldots \cdot P_{\ell-1}e_{y_\ell}\cdot P_\ell$ in $G^-\cup E^+=G$,
where $u_{y_i}v_{y_i}=e_{y_i}$, $v_{y_0}:=s$ and $u_{y_{\ell+1}}:=t$,
such that every $v_{y_{i}}\to u_{y_{i+1}}$ path in $G^-$ (for $i=0,\ldots,\ell$), in particular $P_{i}$, is fully contained
in $G^-[V_{s,y_{i}}]$.

The above reduces our problem to computing, for all $i=0,\ldots,\ell$,
some $p\to q$ path in $G^-[V_{s,y_i}]$, where $p:=v_{y_{i}}$ and $q:=u_{y_{i+1}}$.
We accomplish this for each $i$ separately.
Let $X,Y\in\sccs$ be such that $p\in X$ and $q\in Y$.
Since every $p\to q$ path in $G^-$ lies in
$G^-[V_{s,y_{i}}]$, by Lemma~\ref{l:path-in-condensation} we can find an $X\to Y$ path $P'$
in $G^-/\sccs$ in $O(|V_{s,y_{i}}|\cdot n^\rho)$ time.
We now discuss how such a path can be lifted to an actual $p\to q$
path in $G$.

From the auxiliary data structures of Lemma~\ref{l:decscc-extension}
we can obtain, for each edge $e'$ of $P'$ connecting some $A,B\in \sccs$,
an edge $ab=e\in E(G^-)$ such that $a\in A$ and $b\in B$.
As a result, from $P'$ we obtain a sequence of edges $a_lb_l$, $l=1,\ldots,g$, such that
for each $l=0,\ldots,g+1$, (1) $b_{l-1}$ and $a_l$ lie in the same strongly connected
component of $G^-$, (2) $a_l$ and $b_l$ lie in distinct strongly connected
components of $G^-$ (where $b_0:=p$ and $a_{g+1}:=q$).
In order to convert this sequence into a $p\to q$ path in~$G^-$,
we query the data structure $\decscc$ for the respective $b_{l-1}\to a_l$ paths.
Note that all the requested paths lie within single but distinct strongly connected components of $G^-$.
As the paths returned by $\decscc$ are simple and pairwise disjoint,
and all are fully contained within $G^-[V_{s,y_{i}}]$,
the total time needed to compute them all will be $O(|V_{s,y_i}|\cdot n^\rho)$.

We conclude that the total time needed to report an $s\to t$ path
is $O\left(n^{\rho}\cdot \sum_{i=0}^\ell |V_{s,y_i}|\right)$, which,
since the sets $V_{s,i}$ are pairwise disjoint,
is $O(n^{1+\rho})$.

We stress that the above algorithm is deterministic, modulo the operation
of data structures $\dtc$ and $\mathcal{Q}$, whose output is also
unique with high probability.
As a result, we can report paths against an adaptive adversary
if and only if the paths within strongly connected components of $G^-$
can be reported against an adaptive adversary.

Finally, note that a new data structure $\decscc$ is initialized once per every $F$ insertions.
Thus, the amortized update time can be bounded as $\Ot(n^{1+\rho}+nF+T(n,m)/F)$.
To obtain Theorem~\ref{t:path-reporting}, it is enough
to set $F=\sqrt{T(n,m)/n}$.

\begin{corollary}
  There exist data structures supporting path queries between arbitrary vertices of a \emph{fully dynamic} graph
  with the following amortized update time and worst-case query time bounds:
  \begin{enumerate}
    \item $\Ot(m^{1/2}\cdot n^{5/6+o(1)}+n^{1+\rho})=O(n^{1.834})$ against an adaptive adversary,
    \item $\Ot(n^{1+\rho})=O(n^{1.529})$ against an oblivious adversary.
  \end{enumerate}
\end{corollary}
\begin{proof}
  For the former item,
  there exists a \emph{deterministic} decremental strongly connected components data structure
  with total update time $mn^{2/3+o(1)}$ that supports path queries within an SCC
  in $O|P|\cdot n^{o(1)}$ time, where $P$ is the reported simple path~\cite{BernsteinGS20}.
  For the latter, we note that the near-optimal data structure of~\cite{BernsteinPW19} (Theorem~\ref{t:decscc}) can be
  also adjusted to support path queries within a strongly connected component, but only against an oblivious adversary.
\end{proof}

\section{Fully dynamic reachability tree reporting in general digraphs}\label{s:tree-general}
In this section, we show the following theorem.

\ttreereporting*
Our overall strategy will be to generalize the approach of Section~\ref{s:dag-tree} to general directed graphs.
Unfortunately, Observation~\ref{o:outtree} does not hold for general graphs:
given some source $s\in V$,
choosing an arbitrary incoming edge (e.g., that with the minimum label) from each vertex reachable from $s$
might lead to a disconnected (in the undirected sense) graph containing cycles.

In order to deal with this problem, we could apply Observation~\ref{o:outtree} to the condensation
$G/\sccs$ (with the source vertex set to the SCC containing $s$). The obtained
out-tree $T'$ in $G/\sccs$ could be then extended to an out-tree $T$ in $G$ spanning the vertices
reachable from $s$ in two steps. First, expand each vertex of $T'$ (i.e., a strongly connected
component $S$ reachable from $s$) into
a sparse strongly connected subgraph of $S$, e.g., consisting of a pair of reachability
trees from/to a single vertex in~$S$.
Then, compute a single-source reachability tree from $s$ in the obtained subgraph of $G$
using any graph search procedure in $O(n)$ time.

However, it is not clear how to efficiently operate on the condensation $G/\sccs$ (which is a DAG)
of a decremental graph $G$
using the algebraic data structures of Theorem~\ref{t:dyn-tr}~and~Lemma~\ref{l:dyn-tr-general}.
Recall that when the strongly connected components split, the condensation undergoes vertex splits, and each
vertex split (revealed online) might require $\Theta(n)$ edges changing endpoints.

We will nevertheless avoid this problem. Let us start with the following technical lemma.
\begin{lemma}\label{l:reach-subgraph}
  Let $\pi:V\to [1,n]$ be such that: (a) if $x,y\in V$ lie in distinct strongly connected
  components of $G$ and there exists an $x\to y$ path in $G$, then $\pi(x)<\pi(y)$,
  and (b) if $x$ and $y$ are strongly connected, then $\pi(x)=\pi(y)$.

  Let $s\in V$ and let $V_s\subseteq V$ be the vertices reachable from $s$ in $G$.
  Moreover, let $\sccs_s$ be the set of strongly connected components reachable from $s$ in $G$.
  Let $\sccs_s'$ be an arbitrary subset of $\sccs_s$.

  Suppose $H\subseteq G[V_s]$ with $V(H)=V_s$ satisfies the following:
  \begin{enumerate}[label=(\arabic*)]
    \item for each $t\in V_s\setminus\{s\}$, $H$ contains an edge $vt$ such that $\pi(v)$ is minimal possible
      among $v\in \bigcup \sccs_s'$.
    \item for all $(X,Y)\in (\sccs_s\setminus \sccs_s')\times \sccs_s$, $H$ contains an edge $xy\in E(G)\cap (X\times Y)$,
      if such an edge exists.
    \item for each $S\in \sccs_s$, $H[S]$ is strongly connected.
  \end{enumerate}
  Then, all vertices of $V_s$ are reachable from $s$ in $H$.
\end{lemma}
\begin{proof}
  By item~(3), it is enough to prove that in the condensation $H/\sccs_s$, every $Y\in \sccs_s$
  is reachable from the vertex (component) $S^*$ containing $s$.
  Since $H/\sccs_s$ is a DAG, by Observation~\ref{o:outtree}, we only need to
  prove that every $Y\in \sccs_s\setminus\{S^*\}$ has an incoming edge in that graph.

  Note that $Y$ is reachable from $S^*$ in $G[V_s]/\sccs_s$.
  Consequently, $Y$ has an incoming edge from some $X\in \sccs_s$, $X\neq Y$, in $G[V_s]/\sccs_s$.
  If there exists such an $X$ that $X\notin \sccs_s'$, then by item~(2), we also have an edge $xy\in E(G)\cap (X\times Y)$
  in $H$. So, then in $H/\sccs_s$ indeed $Y$ has an incoming edge.

  So assume that for all $X$ such that $Y$ has an incoming edge $XY$ in $G[V_s]/\sccs_s$, $X\in \sccs_s'$.
  There exists an edge $xy\in E(G)\cap (X\times Y)$ and we have $x\in (\bigcup\sccs_s')\setminus Y$
  and $\pi(x)<\pi(y)$.
  Consequently, by item~(1), $H$ contains an edge $vy$ with $v\in \bigcup\sccs_s'$ such that
  $\pi(v)\leq \pi(x)<\pi(y)$.
  By $\pi(v)<\pi(y)$, $v$ is in a different strongly connected component than $y$ and $v$
  is reachable from $s$ as well.
  This proves that in $H/\sccs_s$ the vertex $Y$ indeed has an incoming edge.
\end{proof}
Generally speaking, the query procedure will construct a subgraph $\mathcal{H}\subseteq G$ of moderate size satisfying
the requirements of~Lemma~\ref{l:reach-subgraph} and then produce a single-source
reachability tree from $s$ in $H$ (and thus also in $G$) in $O(|E(\mathcal{H})|)$ time.
To allow constructing such a subgraph we will require a number of components.

We reuse much of the notation and developments from Section~\ref{s:reach-general}.
Again, the algorithm will operate in phases spanning $F$ edge insertions.
The data structure $\decscc$, accompanied with the auxiliary data structures of Lemma~\ref{l:decscc-extension},
is reinitialized at the beginning of each phase.
Recall that $E^+$ denotes the set of edges inserted in the current phase, and
$G^-$ is the graph from the beginning of the phase minus the edges deleted in the current phase,
so that at all times the current graph $G$ satisfies $G=G^-\cup E^+$.
Moreover, similarly as in Section~\ref{s:reach-general}, we use a pair of data structures
$\mathcal{Q},\dtc$ of Theorem~\ref{t:dyn-tr}~and~Lemma~\ref{l:dyn-tr-general} respectively.
Both these data structures maintain the graph $G^-$ when a phase proceeds, and are passed
the edge insertions of $E^+$ as late as when a new phase is started.

We will maintain a labeling $\pi:V\to [1,n]$ based on the topological labels $\pi$ of the strongly connected components
$\sccs$ of $G^-$, as given by Lemma~\ref{l:decscc-extension} -- we set $\pi(v):=\pi(S)$ if $v\in S\in\sccs$.
\newcommand{\ivals}{\mathcal{I}}

Finally, for $\Delta$ to be chosen later, we maintain $q=\lceil n/\Delta \rceil$ data structures $\dtc_1,\ldots,\dtc_q$
of Theorem~\ref{t:dyn-tr}. Recall that in Section~\ref{s:dag-tree}, $\dtc_i$ was responsible for detecting
paths whose penultimate vertex had its label in $[(i-1)\cdot \Delta+1,i\cdot \Delta]$:
this was possible since the $V\times V'$ layer of the underlying graph~$G_i$ contained only
edges $vv'$ where $v\in [(i-1)\cdot \Delta+1,i\cdot \Delta]$.
Here, we proceed a bit differently: the data structure $\dtc_i$ will be responsible for
detecting paths whose penultimate vertex's \emph{topological label} $\pi(y)$ will
\emph{surely remain} in the interval $[(i-1)\cdot \Delta+1,i\cdot \Delta]$ till the end of the current phase.
In particular, not all possible penultimate vertices will be assigned to one of $\dtc_i$.
For example, it may happen
that for some vertex, the range of its possible topological labels after the future updates of the phase
is too large.

More formally, for each $i=1,\ldots,q$, we maintain a \emph{growing} subset $Y_i\subseteq V$
such that the underlying graph $G_i=(V\cup V'\cup V'',E')$ of $\dtc_i$ satisfies
$uv\in E'$ and $u'v''\in E'$ for each $uv\in E(G^-)$, and moreover $yy'\in E'$
if and only if $y\in Y_i$.
$Y_i$ will contain the vertices whose topological label will stay in the interval
$[(i-1)\cdot \Delta+1,i\cdot \Delta]$ till the end of the current phase.
By proceeding identically as in the proof of Lemma~\ref{l:detect-interval},
one shows that a $u\to v''$ path exists in $G_i$ if and only
if there exists a $u\to v$ path in $G^-$ whose penultimate vertex is in $Y_i$.

Observe that when a new element $z$ gets added to $Y_i$, we can update
$\dtc_i$ in $O(n^{1+\rho})$ time -- this amounts to processing a single edge insertion
$zz'$ to $G_i$.
Similarly as in Section~\ref{s:dag-tree}, every edge update to $G$ can be translated to two edge updates in each $\dtc_i$.
Analogously as in Section~\ref{s:reach-general}, each edge deletion to $G$ is immediately passed to each of $\dtc_i$,
whereas the insertions are executed on these data structures only when the phase ends.

We will maintain the following invariant: for any $S\in\sccs$ and $i=1,\ldots,q$, if $[\pi(S),\pi(S)+|S|-1]\subseteq [(i-1)\cdot \Delta+1,i\cdot \Delta]$
then $S\subseteq Y_i$, and otherwise $S\cap Y_i=\emptyset$.
Note that by the properties of $\pi$ from Lemma~\ref{l:decscc-extension},
if $S\subseteq Y_i$ holds at some point for $S\in\sccs$, $v\in Y_i$ will
hold for all $v\in S$ till the end of the phase. This is because $\pi(v)\in  [(i-1)\cdot \Delta+1,i\cdot \Delta]$ will hold.

We call the strongly connected components $S\in \sccs$ such that
we have $S\cap Y_i=\emptyset$ for all $i=1,\ldots,q$ \emph{special}.
The above invariant will be maintained as follows: when some special component $S$ breaks
into smaller components $S_1,\ldots,S_k$, we check for each of them
if it is special. If $S_j$ is not special, i.e.,
 $[\pi(S_j),\pi(S_j)+|S_j|-1]\subseteq [(i-1)\cdot \Delta+1,i\cdot \Delta]$
 for some $i$, we add the elements of $S_j$ to $Y_i$.
 This requires $|S_j|$ edge insertions issued to $\dtc_i$.
 Afterwards, none of the vertices of $S_j$ will get inserted to any other $Y_l$ anymore in the current phase.

Note that through initialization and all component splits in a phase, the cost of maintaining the invariant
is $O(n\cdot n^{1+\rho})=O(n^{2+\rho})$ since each vertex is added
to some single $Y_i$ at most once.
For the same reason, emptying the sets $Y_i$ once the current phase is
finished (we require them empty when the next phase starts) can be performed in $O(n^{2+\rho})$ time as well.

\begin{observation}\label{o:special}
  At all times, there exist $O(n/\Delta)$ special strongly connected components $S\in\sccs$.
\end{observation}
\begin{proof}
  Let $j$ be such that $\pi(S)\in [(j-1)\cdot \Delta+1,j\cdot\Delta]$.
  Since $S\cap Y_j=\emptyset$,
  $[\pi(S),\pi(S)+|S|-1]\not\subseteq {[(j-1)\cdot \Delta+1,j\cdot \Delta]}$
  which implies that $\pi(S)+|S|-1>j\cdot \Delta$.
  Hence, $j\cdot \Delta\in [\pi(S),\pi(S)+|S|-1]$.
  Since there are $O(n/\Delta)$ numbers of the form $j\cdot \Delta$
  in $[1,n]$, and the intervals $[\pi(S'),\pi(S')+|S'|-1]$
  and $[\pi(S''),\pi(S'')+|S''|-1]$ are disjoint for any $S',S''\in\sccs$, $S'\neq S''$,
  only $O(n/\Delta)$ strongly connected components $S$
  may satisfy $j\cdot \Delta\in [\pi(S),\pi(S)+|S|-1]$ for some $j$.
\end{proof}

Having defined all the needed components, we now describe how to construct a desired
subgraph $\mathcal{H}\subseteq G$.
First of all, we compute the partition $V_{s,0},\ldots,V_{s,k}$ and the tree structure $J$
of Lemma~\ref{l:query-partition}, which takes $O(n^{1+\rho}+nF)$ time.
The subgraph $\mathcal{H}$ will consist of the edges $E^+=\{u_1v_1,\ldots,u_kv_k\}$ and
the union of individual subgraphs $\mathcal{H}_i\subseteq V_{s,i}$ for all $i=0,1,\ldots,k$
(recall that $v_0=s$ in Lemma~\ref{l:query-partition}).
Due to the existence of the tree structure $J$ on the indices $i$ with $V_{s,i}\neq\emptyset$,
it is enough to guarantee that each $\mathcal{H}_i$
contains a reachability tree from $v_i$ in $G^-[V_{s,i}]$.

Let $\sccs_i$ contain the subset of the strongly connected components of $\sccs$ of $G^-$
contained in $G^-[V_{s,i}]$.
For each $Z\in \sccs_i$ that is special, and all other $X\in \sccs_i\setminus\{Z\}$,
if there exists an edge $XZ$ in $G^-[V_{s,i}]/\sccs$, we include in $\mathcal{H}_i$
the edge $e_{X,Z}\in E(G^-)$,
as defined and maintained by Lemma~\ref{l:decscc-extension}.

For each vertex $t\in V_{s,i}\setminus\{v_i\}$,
we find the smallest $j$ (if one exists) such that a path $v_i\to t''$
exists in~$G_{j}$.
Note that such a $j$ can be found using at most $q=O(n/\Delta)$ reachability
queries to the data structures $\dtc_1,\ldots,\dtc_q$.
The cost of these queries is clearly $O((n/\Delta)\cdot n^{\rho})$.
We then add to $\mathcal{H}_i$ all the $t$'s incoming
edges (in $G^-[V_{s,i}]$) from the vertices of the set $Y_j$, whose size is no more than $\Delta$.

Finally, for each $Z\in \sccs$ (regardless of whether it is special or not),
using the assumed decremental strongly connected components data structure $\decscc$ maintaining $G^-$, we compute a strongly connected subgraph $Z'$ of $G^-[Z]$ with $V(Z')=Z$ in $O(|Z|\cdot n^{(1+\rho)/2})$
time and add it to $\mathcal{H}_i$.

The following lemma proves the above construction correct.
\begin{lemma}\label{l:subgraph-corect}
  Let $V_{s,i}\neq\emptyset$. Then, in the subgraph $\mathcal{H}_i$, $v_i$ can reach every vertex of $V_{s,i}$.
\end{lemma}
\begin{proof}
  We prove that $\mathcal{H}_i$ satisfies the properties (1)-(3) from Lemma~\ref{l:reach-subgraph}
  applied to $G:=G^-[V_{s,i}]$, $s:=v_i$, $H:=\mathcal{H}_i$, $\sccs_s:=\sccs_i$, and $\sccs_s'$ equal
  to the \emph{non-special} strongly components in $\sccs_i$.
  Then, the lemma will follow by Lemma~\ref{l:reach-subgraph}.

  Let $t\in V_{s,i}\setminus\{v_i\}$. For contradiction, suppose that property~(1) is not satisfied,
  i.e., $\mathcal{H}_i$ does not contain an edge $vt$ such that $v$ is in a non-special strongly connected
  component $Z\in \sccs_i$ and $\pi(v)$ is minimal.
  Since $Z$ is non-special, by the maintained invariant we have $Z\subseteq Y_l$ for some $l$.
  Moreover, since $v$ is reachable from $v_i$ in $G^-[V_{s,i}]$, there exists a path $v_i\to t$
  in $G^-[V_{s,i}]$ whose penultimate vertex is $v$.
  As a result, there exists a path $v_i\to t''$ in $G_l$, i.e., the query about a $v_i\to t''$ path
  in~$\dtc_l$ returns true.
  So, the minimal $j$ such that a path $v_i\to t''$ exists in $G_j$ satisfies $j\leq l$.
  If $j=l$, then all edges $wt\in E(G^-)$ with $w\in Y_l$, in particular $vt$, are added
  to $\mathcal{H}_i$, which contradicts our assumption.
  On the other hand, if $j<l$, then there exists a $v_i\to t$ path in $G^-[V_{s,i}]$ whose
  penultimate vertex $w$ lies in $Y_j$. But since $j<l$, $wt\in E(G^-[V_{s,i}])$, $\pi(w)<\pi(v)$ and $w$ lies
  in a non-special strongly connected component as well, so $\pi(v)$ was not minimal possible -- a contradiction.

  Property~(2) from Lemma~\ref{l:reach-subgraph} follows easily by construction: for each special
  strongly connected component $X$, we add to $\mathcal{H}_i$ an outgoing edge to every other component of $G^-[V_{s,i}]$ (if it exists).

  Property~(3) also follows trivially by construction.
\end{proof}

\begin{lemma}\label{l:subgraph-size}
  The total time needed to compute the subgraph $\mathcal{H}$ is $O(n^{2+\rho}/\Delta+n\Delta+n^{(3+\rho)/2}+nF)$.
\end{lemma}
\begin{proof}
  First, recall that computing the partition $V_{s,0},\ldots,V_{s,k}$ takes $O(n^{1+\rho}+nF)$ time.

  Let $r_i$ denote the number of special strongly connected components in $\sccs_i$.
  Then, the total time needed to compute $\mathcal{H}_i$ is:
  \begin{equation*}
    O\left(r_i\cdot |V_{s,i}|+|V_{s,i}|\cdot \left(\frac{n}{\Delta}\cdot n^{\rho}+\Delta\right)+|V_{s,i}|\cdot n^{(1+\rho)/2}\right).
  \end{equation*}
  Summing through all $i=0,\ldots,k$, and taking into account that there are $O(n/\Delta)$ special strongly connected components in total (see Observation~\ref{o:special}), we bound the time as follows:
  \begin{equation*}
    O\left(\left(\frac{n^{1+\rho}}{\Delta}+n^{(1+\rho)/2}+\Delta\right)\cdot \sum_{i=1}^k|V_{s,i}|+\sum_{i=1}^k|V_{s,i}|\cdot r_i\right)=O\left(\frac{n^{2+\rho}}{\Delta}+n\Delta+n^{(3+\rho)/2}+\frac{n^2}{\Delta}\right).\qedhere
  \end{equation*}
\end{proof}
  Note that the time bound from Lemma~\ref{l:subgraph-size} also bounds the size of $\mathcal{H}$,
  and thus the query time.

\begin{lemma}
  The amortized update time is $O\left(T(n,m)/F + n^{2+\rho}/\Delta + n^{2+\rho}/F\right)$.
\end{lemma}
\begin{proof}
  Initializing the data structure $\decscc$
  along with the auxiliary data of Lemma~\ref{l:decscc-extension}
  happens once per $F$ insertions, this gives $O(T(n,m)/F)$ amortized time per update.
  Since each edge update is passed to all the data structures $\dtc,\dtc_1,\ldots,\dtc_q$,
  which gives $O(n^{2+\rho}/\Delta)$ amortized cost per update.
  As discussed earlier, handling changes to the sets $Y_i$ (and also their initialization)
  in data structures $\dtc_1,\ldots,\dtc_q$
  costs $O(n^{2+\rho})$ time per phase, which gives $O(n^{2+\rho}/F)$ amortized time per update.
\end{proof}

To balance the update and query time, we set $\Delta=F=\max\left(\sqrt{\frac{T(n,m)}{n}},n^{(1+\rho)/2}\right)$.
Again, all the components except the subgraphs certifying strong connectivity
inside the components of $\sccs$ are deterministic, so the algorithm works against an
adaptive adversary if and only if these subgraphs are reported against an adaptive adversary.

\begin{corollary}
  There exist data structures supporting single-source reachability tree queries on a \emph{fully dynamic} graph
  with the following amortized update time and worst-case query time bounds:
  \begin{enumerate}
    \item $\Ot(m^{1/2}\cdot n^{5/6+o(1)}+n^{(3+\rho)/2})=O(n^{1.834})$ against an adaptive adversary,
    \item $\Ot(n^{(3+\rho)/2})=O(n^{1.765})$ against an oblivious adversary.
  \end{enumerate}
\end{corollary}
\begin{proof}
  For the former item,
  we note that the \emph{deterministic} decremental strongly connected components data structure
  with total update time $mn^{2/3+o(1)}$~\cite{BernsteinGS20}
  is actually capable of finding the first edge on some simple $s\to t$ path, where $s$ and $t$ are strongly connected,
  in $n^{o(1)}$ \emph{worst-case time}.
  As a result, for any strongly connected component $S$, we can find a sparse strongly connected subgraph
  of $G[S]$ spanning vertices $S$ in $|S|\cdot n^{o(1)}$ time as follows.
  Pick some root $r\in S$. The subgraph will consist of an out-tree from $r$, and an in-tree to $r$,
  both within $G[S]$.
  Let us focus on finding an in-tree $T_S$, since an out-tree can be found by proceeding identically
  on the reverse graph.
  Initially, $T_S=(\{r\},\emptyset)$. While there exists some $x\in S\setminus V(T_S)$, pick
  an arbitrary such $x$.
  Issue a query about an $x\to r$ path in $G[S]$ to $\decscc$ -- the query will produce, one by one,
  the subsequent edges $e_1=x_1y_1,e_2=x_2y_2,\ldots,$ of some $x\to r$ path,
  each in $n^{o(1)}$ worst-case time.
  We stop when the algorithm outputs, for the first time, an edge such $e_j$ such that $y_j\in V(T_S)$.
  Observe that then the algorithm terminates in $j\cdot n^{o(1)}$ time.
  Moreover, by connecting the path $e_1\ldots e_j$ to the tree $T_S$, we make
  $T_S$ span $j$ more vertices.
  As a result, $V(T_S)$ will finally grow to size $|S|$ and thus constructing
  $T_S$ will require $|S|\cdot n^{o(1)}$ time, which is $O(|S|\cdot n^{(1+\rho)/2})$, as desired.

  Consider item 2.
  The near-optimal data structure of~\cite{BernsteinPW19} (Theorem~\ref{t:decscc})
  actually maintains, for each strongly connected component $S\in\sccs$, a spanning in-tree and a spanning out-tree.
  Combined, those can serve as a sparse subgraph of $G[S]$ spanning $S$ that is strongly connected.
  Unfortunately, these maintained trees might reveal the random choices made by the data structure.
  So the tree reporting in this case works against an oblivious adversary only.
\end{proof}

\section{Incremental reachability}\label{s:incremental-reach}
In this section we show deterministic incremental reachability algorithms with subquadratic update and query bounds.
We start with the following technical lemma.

\begin{lemma}\label{l:trees-recompute}
  Suppose for each $s\in V$ we are given a reachability tree $T_G(s)$ from $s$ in $G$.
  Let $W$ be some subset of $V$ of size $O(n^\alpha)$, where $\alpha\in [0,1]$.
  Let $E^+\subseteq V\times W$.
  Then, one can construct the reachability trees $T_{G'}(s)\subseteq G\cup E^+$ for all $s\in V$ deterministically
  in $\Ot(n^{\omega(1,\alpha,1)})$ time.
\end{lemma}
\begin{proof}
  Denote by $G^+$ the graph $G\cup E^+$.
  Let us first compute the strongly connected components of $\sccs$ of $G^+$.
  This takes $O(n^2)$ time.
  The overall idea is to compute the reachability trees in $G^+/\sccs$
  and then expand them in $O(n^2)$ time to form reachability trees in $G^+$.

  Indeed, suppose that we are given a reachability tree $T_{G^+/\sccs}(C)$ from some SCC $C\in \sccs$.
  Each edge $XY$ of $T_{G^+/\sccs}(C)$ corresponds to some edge $xy=e_{X,Y}\in E(G^+)$
  with $x\in X$ and $y\in Y$.
  For each $S\in\sccs$, fix some arbitrary vertex $v_S\in S$.
  Clearly, in $O(n^2)$ total time we can compute, for all $S\in \sccs$,
  a reachability out-tree $T_{G^+[S]}(v_S)^\text{out}$, and
  a reachability in-tree $T_{G^+[S]}(v_S)^\text{in}$.
  Now, to expand $T_{G^+/\sccs}(C)$ into a reachability tree $T_{G^+}(s)$ for
  an arbitrary $s\in C$, we can simply run depth-first search on
  the graph
  \begin{equation*}
    \bigcup_{C\in\sccs}\left( T_{G^+[C]}(v_C)^\text{out}\cup T_{G^+[C]}(v_C)^\text{in}\right)\cup E(T_{G^+/\sccs}(C)),
  \end{equation*}
  which clearly has $O(n)$ size. Through all $C\in\sccs$ and $s\in C$, this expansion takes $O(n^2)$ time.

  The graph $G^+/\sccs$ is a DAG. Therefore, by Observation~\ref{o:outtree}, in order
  to compute a reachability tree $T_{G^+/\sccs}(C)$ from $C\in\sccs$,
  it is enough enough to pick, for each $X\in \sccs\setminus \{C\}$,
  any incoming edge $e_{C,X}=YX\in E(G^+/\sccs)$ such that $Y$ is reachable from $C$ in $G^+/\sccs$.
  We achieve that as follows.

  First, we compute, for each $C\in\sccs$, a reachability tree $T_{G/\sccs}(C)$ from $C$
  in the graph $G$ with every SCC of $G^+$ contracted.
  Observe that the edges of $T_{G/\sccs}(C)$ can be picked from, say, $T_{G}(v_C)$
  in $O(n)$ time. Through all SCCs $C$, this amounts to $O(n^2)$ time.
  Note that if for some $X\in\sccs$ we have $X\in V(T_{G/\sccs}(C))$, then
  $e_{C,X}$ can be taken as the incoming edge of $X$ in $T_{G/\sccs}(C)$.

  But it might happen that despite $X\notin V(T_{G/\sccs}(C))$, $X$ is reachable
  from $C$ in $G^+/\sccs$. Observe that this is only possible if
  every $C\to X$ path in $G^+/\sccs$ goes through a vertex $Q\in\sccs$
  such that $Q\cap W\neq \emptyset$, or more specifically, if every path from a vertex of $C$ to a vertex of $X$
  in $G^+$ goes through an edge of $E^+$.
  We now discuss how to find the edges $e_{C,X}$ that fall into this case.

  Let $\mathcal{T}=\{C\in\mathcal{S}:C\cap W\neq\emptyset\}$ and set $\ell=|\mathcal{T}|$.
  Let $L$ be a boolean matrix with rows $\sccs$ and columns~$\mathcal{T}$ such that
  \begin{equation*}
    L_{C,D}=[\text{there exists an }C\to D\text{ path in }G^+/\sccs].
  \end{equation*}
  We now show that the matrix $L$ can be computed in $O(n^{\omega(1,\alpha,1)})$ time.

  Let the boolean matrix $A$ with rows and columns $\sccs$ be such that
  \begin{equation*}
    A_{X,Y}=[\text{there is an }X\to Y\text{ path in }G/\sccs\text{ or }xy\in E^+\text{ for some }x\in X\text{ and }y\in Y] \text{ for all }X,Y\in \sccs.
  \end{equation*}
  Note that the matrix $A$ can be filled by first setting $A_{X,Y}=[Y\in T_{G/\sccs}(X)]$  for all $X,Y$.
  Then it can be easily updated to capture the edges $E^+$ in $|E^+|=O(n\cdot |W|)=O(n^2)$ time.

  Now, to obtain the matrix $L$, we compute the boolean product $A[\sccs,\mathcal{T}]\cdot A[\mathcal{T},\mathcal{T}]^{2\ell+1}$,
  where by $A[I,J]$ we denote the submatrix of $A$ including only the rows of $I$ and columns of $J$.
  This can be done in $O(n^{\alpha\cdot \omega}\log{n}+n^{\omega(1,\alpha,1)})=\Ot(n^{\omega(1,\alpha,1)})$
  time using fast matrix multiplication and repeated squaring for exponentiation.  Let us prove this is correct.

  Define $M = A[\sccs,\mathcal{T}]\cdot A[\mathcal{T},\mathcal{T}]^{2\ell+1}$.  We need to show
  $L_{C,D}$ is true if and only if $M_{C,D}$ is, for all $C \in \sccs$ and $D \in \mathcal{T}$.
  First, observe that $A[\mathcal{T},\mathcal{T}]^{2\ell+1}_{C,D}$ is true if and only
  if there exists a $C\to D$ path in $G^+/\sccs$ since any $C\to D$ simple path in $G^+/\sccs$
  can be divided into at most $2\ell+1$ maximal subpaths
  either fully contained in $G/\sccs$, or consisting of an edge corresponding
  to some element of~$E^+$.
  This follows from the fact that the number of edges from $E^+$ on such
  a simple path is no larger than $|\mathcal{T}|=\ell$, and there can be no
  more than $\ell+1$ subpaths of the other type.
  Second, note that each $C\to D$ path in $G^+/\sccs$ contains a prefix
  fully contained in $G/\sccs$ whose other endpoint is in $\mathcal{T}$.

  Let the boolean matrix $B$ with rows $\mathcal{T}$ and columns $\sccs$ be such that
  \begin{equation*}
    B_{X,Y}=[A_{X,Y}=1\text{ and } X\neq Y]\text{ for all }X\in \mathcal{T}, Y\in \sccs.
  \end{equation*}

  Finally, we compute the boolean product $L\cdot B$ along with the corresponding witnesses.
  Specifically, for all $X,Y\in \sccs$, if $\left(L\cdot B\right)_{X,Y}=0$
  then $w_{X,Y}=\perp$, and otherwise, $w_{X,Y}\in \mathcal{T}$ is such
  that $L_{X,w_{X,Y}}=1$ and $B_{w_{X,Y},Y}=1$, which implies $w_{X,Y}\neq Y$.
  The witnesses for boolean matrix multiplication can be computed deterministically
  in nearly matrix multiplication time~\cite{AlonGMN92}.
  This applies also to rectangular matrix multiplication (see, e.g.,~\cite{GrandoniILPU21}).
  As a result, computing all $w_{X,Y}$ takes $\Ot(n^{\omega(1,\alpha,1)})$ time.

  Finally, for each $e_{C,X}$, where $C,X\in\sccs$, that has not been yet set,
  we will set it to the incoming edge of $X$ in $T_{G/\sccs}(w_{C,X})$ provided that $w_{C,X}\neq\perp$.

  Indeed, note that $w_{C,X}\neq\perp$ if and only if there
  exists a path $C\to w_{C,X}$ in $G^+/\sccs$, $w_{C,X}\neq X$, and there exists either (a) a $w_{C,X}\to X$ path in $G/\sccs$,
  or (b) an edge $xy\in E^+$ such that $x\in w_{C,X}$ and $y\in X$.
  Clearly, in both cases the respective paths certify that a path $C\to X$ exists in $G^+/\sccs$.
  In the other direction, if a $C\to X$ path in $G^+/\sccs$ goes through some last
  $Y\in \mathcal{T}$, and $Y$ is the final vertex of $\mathcal{T}\setminus\{X\}$ on this path,
  then $L_{C,Y}=1$ and $A[Y,X]=1$, and thus $w_{C,X}\neq\perp$.
  Finally, note that the incoming edge of $X$ in $T_{G/\sccs}(w_{C,X})$ comes
  from a component $Z\in \sccs$, $Z\neq X$ that $w_{C,X}$ can reach in $G/\sccs$, and
  thus from a component that $C$ can reach in $G^+/\sccs$, as desired.
\end{proof}

The incremental reachability data structures in this section (as well as the shortest paths
data structures in the following sections) will operate in phases of updates.
In the following, let $F=n^\alpha$
denote the phase length, to be set later. Let $G_0$ denote the graph $G$
at the beginning of the phase.
The $i$-th update of the phase ($i=1,\ldots,F$) inserts some incoming edges to a vertex $v_i$.
Let~$E_i$ denote these edges.
Define $E_k^+=\bigcup_{i=1}^k E_i$.

At all times, we store a collection of reachability trees $T_{G_0}(s)$ for all $s\in V$.
When a new phase starts, let $G_0'$ and $E'$ denote the graph $G$
at the beginning of the previous phase, and the set of edges inserted in the previous phase,
respectively.
The new reachability trees $T_{G_0}(s)$ from all $s\in V$
are recomputed based on the reachability
trees $T_{G_0'}(\cdot)$ and the edges $E'$.
Since the edges $E'$ have at most $F$ distinct heads, by Lemma~\ref{l:trees-recompute},
this takes $O(n^{\omega(1,\alpha,1)})$ time.
Since such a recomputation happens once every $n^\alpha$ updates,
the amortized per-update time spent on this is $O(n^{\omega(1,\alpha,1)-\alpha})$.

The amortization in all the data structures in this section (and in the following sections)
will come only from performing a certain costly rebuilding step that takes $T(n)$ worst-case time once per
a phase of updates; all the remaining computation will have desired worst-case bounds.
There is a well-known standard technique (also used in the previous dynamic algebraic graph algorithms, e.g.~\cite{BrandNS19, DemetrescuI05, Sankowski04}) for converting such
amortized bounds into worst-case bounds by maintaining two copies of the data structure
that switch their roles every $F/2$ updates. One copy is for handling
at most $F/2$ updates and answering queries, and the other is being gradually
reinitialized in chunks of $\Theta(T(n)/F)$ time.
Thus, in the following we will assume this technique is used wherever applicable
and state our bounds as worst-case instead of amortized.

We now describe two different trade-offs for incremental path reporting.

\subsection{Single-source reachability tree reporting}
In this section, we wish to support single-source reachability tree queries.
Suppose $k$ updates have already happened in this phase.

The algorithm is very simple. Let $s$ be a query vertex.
Let $H_s$ be a graph obtained by taking the union of the trees $T_{G_0}(v)$ for
$v\in \{s,v_1,\ldots,v_k\}$ and extending it with edges $E_k^+$.
To obtain a single-source reachability tree from $s$, we simply
perform a graph search on $H_s$.
This takes time linear in the size of $H_s$, which is $O(k\cdot n)=O(nF)=O(n^{1+\alpha})$.
The below lemma proves this correct.
\begin{lemma}\label{l:auxiliary-graph}
  For any $v\in V$, there exists an $s\to v$ path in $G$ if and only if there exists an $s\to v$ path in $H_s\subseteq G$.
\end{lemma}
\begin{proof}
  For the forward direction,  if the $s\to v$ path in $G$ existed at the beginning of the phase then
  the path is contained in the tree $T_{G_0}(s)$ itself and so is in $H_s$.  Otherwise, the path must have
  been created after some updates during the current phase.  Let an inclusion-wise minimal
  (in terms of edges of $E_k^+$) $s\to v$ path in $G$ use
  some $k' \le k$ newly inserted edges (with distinct heads) in some specific order
  $e_{i_1}, e_{i_2}, \ldots, e_{i_{k'}}$.
  We can construct this path also in $H_s$ as follows.  Starting at $s$ we continue on
  $T_{G_0}(s)$ till $u_{i_1}$,  use the edge $e_{i_1}= u_{i_1}v_{i_1}$ (as $E_k^+ \subseteq H_s$)
  and then continue on $T_{G_0}(v_{i_1})$ till we reach $e_{i_2}$ and so on.
  Since $H_s\subseteq G$, the other direction is immediate.
\end{proof}
The first trade-off is summarized as follows.
\begin{lemma}\label{l:reach-tree}
  Let $\alpha\in (0,1)$.
  There exists a deterministic incremental data structure supporting insertions of multiple incoming edges
  of a single vertex in $\Ot(n^{\omega(1,\alpha,1)-\alpha})$ worst-case time, and reporting a single-source reachability tree
  from any vertex in $O(n^{1+\alpha})$ worst-case time.
  In particular, for $\alpha=\rho$, both updates and queries can be performed in $\Ot(n^{1+\rho})=O(n^{1.529})$ worst-case time.
\end{lemma}

\subsection{Path reporting with linear cost or better}\label{s:reach-path}
In this section, we only support less general \emph{single-edge} updates,
i.e., $E_i=\{e_i\}$, where $e_i=u_iv_i$.

Let us start by defining a certain family $(P)_{u,v\in V}$ of paths in $G$.
Fix some moment $k$ of the phase.
For any $u,v\in V$, let $l_{u,v}$ be the minimal $l\in \{0,\ldots,k\}$ such that a path $u\to v$ exists in $G_0\cup E_{l}^+$.
If no such $l$ exists, then set $l_{u,v}=\infty$.
The paths $P_{u,v}$ are defined inductively as follows.
If $l_{u,v}=0$, then $P_{u,v}$ is the $u\to v$ path in $T_{G_0}(u)$.
Otherwise, for $l=l_{u,v}$, we have $P_{u,v}=P_{u,u_{l}}\cdot e_l \cdot P_{v_{l},v}$.
Clearly, in this case we have $l_{u,{u_l}}<l$ and $l_{v_{l},v}<l$.

\begin{lemma}\label{l:incr-simple}
  For any $u,v\in V$, $P_{u,v}$ is a simple path.
\end{lemma}
\begin{proof}
  We prove this by induction on $l_{u,v}$. The lemma is clear for $l_{u,v}=0$.
  Suppose $l=l_{u,v}\geq 1$. Then, $P_{u,v}=P_{u,u_{l}}\cdot e_l\cdot P_{v_{l},v}$.
  Since $u_{l}\neq v_{l}$, and by the inductive assumption, $P_{u,v}$ could be non-simple if
  we had $x\in V(P_{u,u_{l}})\cap V(P_{v_{l},v})\neq \emptyset$ for some $x$.
  But then, the $u\to x$ prefix of $P_{u,u_{l}}$ along with the $x\to v$ suffix
  of $P_{v_{l},v}$ would form, concatenated, a $u\to v$ path in $G_0\cup E_{l-1}^+$.
  This contradicts the definition of $l_{u,v}$.
\end{proof}

Apart from the reachability trees in $G_0$ (which encode all paths $P_{u,v}$ with $l_{u,v}=0$), the algorithm will also maintain
some of the paths~$P_{u,v}$ with $l_{u,v}\geq 1$.
We will have the following invariant (I): if for some $u\in V$ and $i=1,\ldots,k$,
$l_{u,u_i}<i$, then the path $P_{u,u_i}$ is stored.
Symmetrically, if for some $v\in V$ and $i=1,\ldots,k$,
$l_{v_i,v}<i$, then the path $P_{v_i,v}$ is stored.
The below lemma shows how these paths can be useful for answering arbitrary path queries.

\begin{lemma}\label{l:concatenate}
  Let $s,t\in V$ be such that $l_{s,t}\geq 1$. Then, a path $s\to t$ exists in $G=G_0\cup E_k^+$ if and only if for
  some $i=1,\ldots,k$ we have $l_{s,u_i}<i$ and $r_{v_i,t}<i$.

  Moreover, if $i$ is minimal such that $l_{s,u_i}<i$ and $r_{v_i,t}<i$, then $P_{s,t}=P_{s,u_i}\cdot e_i\cdot P_{v_i,t}$.
\end{lemma}
\begin{proof}
  The ``$\impliedby$'' direction is trivial. For the ``$\implies$'' part,
  let $i=l_{s,t}\leq k$. It is enough to recall that $P_{s,t}=P_{s,u_i}\cdot e_i\cdot P_{v_i,t}$,
  and in such a case $l_{s,u_i},l_{v_i,t}<i$.
  If for some $j<i$ we had $l_{s,u_j}<j$ and $r_{v_j,t}<j$, then we would have $l_{s,t}\leq j$, a contradiction.
\end{proof}

To allow for easy operation on the stored and reported paths $P_{u,v}$, each
$P_{u,v}$ is implemented using a purely functional concatenable deque with
$O(1)$ worst-case operation costs~\cite{KaplanT99}.
Then, paths can be concatenated in $O(1)$ worst-case time and iterated through in $O(1)$
worst-case time per edge as well (either forward or backwards; this corresponds to repeatedly popping
the first or the last element of the deque).
Note that when a phase starts, we can in fact initialize all the deques $P_{u,v}$
for $l_{u,v}=0$ in $O(n^2)$ worst-case time: for all vertices $v$ in $T_{G_0}(u)$ in the order
of their depth, if $e_{u,v}=zv$ is the incoming edge
of $v$ in $T_{G_0}(u)$, then $P_{u,v}$ can be set to~$P_{u,z}\cdot e_{u,v}$, which
costs a single concatenation.

If the invariant (I) is satisfied, to handle a query $(s,t)$, we first check whether
we already store~$P_{s,t}$. If so, we return it.
Otherwise, we search
for the minimum $i=1,\ldots,k$ such that both paths $P_{s,u_i}$ and $P_{v_i,t}$ are stored.
By Lemma~\ref{l:concatenate}, if no such $i$ is found, then no $s\to t$ path exists in~$G$.
Otherwise, we return $P_{s,u_i}\cdot e_i\cdot P_{v_i,t}$.
Clearly, this path can be assembled using two concatenate operations, and by Lemma~\ref{l:concatenate}
it is equal to $P_{s,t}$. The path $P_{s,t}$, in turn, is simple by Lemma~\ref{l:incr-simple}.
All the subsequent edges of the returned path $P_{s,t}$ can be clearly
returned in $O(1)$ worst-case time per edge.
As a result, the worst-case query cost is $O(k)=O(F)=O(n^{\alpha})$ plus linear in the
number of initial edges of $P_{s,t}$ that we want to report.

Finally, we discuss how the update procedure fixes invariant (I) after inserting an edge $e_i$.
We only need to ensure that we store the paths $P_{u,u_i}$ for all $u\in V$ and $P_{v_i,v}$ for each $v\in V$,
if the respective paths exist.
To this end, it is enough to run the query procedure $O(n)$ times and store
the returned ($O(1)$-space representation of the) path,
without iterating through its individual edges.
Consequently, the worst-case update time is $O(nF)=O(n^{1+\alpha})$.
We have thus proved the following.

\begin{lemma}\label{l:reach-path}
  There exists a deterministic incremental data structure supporting insertions of single edges
  in $\Ot(n^{1+\rho})=O(n^{1.529})$ time and path queries in $O(n^{\rho}+|P|)=O(n^{0.529}+|P|)$
  worst-case time, where $|P|$ is the size of the returned \emph{simple} path.
\end{lemma}

Lemmas~\ref{l:reach-tree}~and~\ref{l:reach-path} combined yield Theorem~\ref{t:incremental-reach-sum}.

\section{Incremental approximate shortest paths}\label{s:incremental-appr}
In this section we consider computing $(1+\eps)$-approximate shortest paths
in the incremental setting.

We will make use of the following result of Zwick~\cite{Zwick02}.
\begin{theorem}\label{t:zwick-appr}
  Let $A = \{a_{i,j}\}$ be an $n\times m$ matrix and let $B = \{b_{i,j}\}$ be an $m\times p$ matrix such that $m=O(n^\alpha)$ and $p=O(n^\beta)$,
  both with real entries from the interval $[1,C]$.
  Then, a \emph{(1+\eps)-approximate min-plus product} $A\star_\eps B=D=\{d_{i,j}\}$ of $A$ and $B$, along with the witnesses, can be
  computed deterministically in 
  $\Ot(n^{\omega(1,\alpha,\beta)}\cdot (1/\eps)\cdot \log(C/\eps))$ time.
  More formally, within this time bound one can compute values $d_{i,j}$ and $w_{i,j}$ such that for all
  $i=1,\ldots,n$, and $j=1,\ldots,p$ we have:
  \begin{equation*}
    \min_{k=1}^m \{a_{i,k}+b_{k,j}\}\leq d_{i,j}=a_{i,w_{i,j}}+b_{w_{i,j},j}\leq (1+\eps)\cdot \min_{k=1}^m \{a_{i,k}+b_{k,j}\}.
  \end{equation*}
\end{theorem}

Let $\eps'$ be the internal error parameter of our data structure to be fixed later as a function of $\eps$ and $n$.
We will use the following analogue of Lemma~\ref{l:trees-recompute}.

\begin{lemma}\label{l:appr-recompute}
  Suppose the matrix $D$ represents $(1+\eps')^k$-approximate all-pairs distances in $G$.
  Let the matrix $P$ contain the respective underlying paths, implemented as concatenable deques as in Section~\ref{s:reach-path}.
  Let~$W$ be some subset of $V$ of size $O(n^\alpha)$, where $\alpha\in [0,1]$.
  Let $E^+$ be some set of edges with real weights in $[1,C]$, and their heads in the set $W$.

  Then one can construct a matrix $D'$ representing $(1+\eps')^{k+O(1)}$-approximate all-pairs
  distances in $G\cup E^+$, along with a matrix of respective (possibly non-simple) paths $P'$ (stored analogously as the matrix $P$), deterministically in $\Ot(n^{\omega(1,\alpha,1)}\cdot (1/\eps')\cdot \log(C/\eps'))$ time.
\end{lemma}
\begin{proof}[Proof sketch.]
  The proof is completely analogous to that of Lemma~\ref{l:trees-recompute} and thus we only provide a sketch.
  Let $D^*$ be the matrix such that $D^*_{u,v}=\min(D_{u,v},\wei_{G\cup E^+}(uv)\}$.
  Let $\eps''=\eps/\Theta(\log{n})$ and $\ell=|W|$.
  First, one computes $(1+\eps')$-approximate distances between the vertices
  $W$ in $G\cup E^+$ by computing the $(2\ell+1)$-st power of the
  matrix $D^*[W,W]$ using the $(1+\eps'')$-approximate min-plus product from Theorem~\ref{t:zwick-appr}
  and repeated squaring.
  Since $\log_2{(2\ell+1)}=O(\log{n})$, this indeed yields
  $(1+\eps')^{k+O(1)}$-approximate distances between the vertices of~$W$ in $G\cup E^+$.
  Finally, the matrix $D'$ is obtained by taking the $(1+\eps')$-approximate min-plus product
  \begin{equation*}
    D^*[V,W]\star_\eps D^*[W,W]^{2\ell+1}\star_\eps D^*[W,V]
  \end{equation*}
  which introduces one more $(1+\eps')$ multiplicative error factor. This takes
  $\Ot(n^{\omega(1,\alpha,1)}\cdot (1/\eps')\cdot \log(C/\eps'))$ time by Theorem~\ref{t:zwick-appr}.

  The paths of the matrix $P'$ can be easily assembled from the paths in $P$
  and witnesses produced when computing the approximate min-plus product in $\Ot(n^2)$ time.
\end{proof}

In the following we will assume that whenever Lemma~\ref{l:appr-recompute} is applied,
we also compute the representation of the respective paths. For simplicity, we will refer to length-path inputs/outputs
of Lemma~\ref{l:appr-recompute} jointly as a \emph{matrix of approximate shortest paths}.

The algorithms in this section will operate in phases spanning $F=\Theta(n^\alpha)$ insertions as well.
At the beginning of each phase, we would like to recompute a matrix $D$ of approximate shortest paths
in~$G_0$ based on an analogous matrix from the previous phase.
However, we cannot simply recompute (approximate) shortest paths using Lemma~\ref{l:appr-recompute}
every $F$ updates: this would accumulate error so that after $q$ recomputations,
the multiplicative error would be $(1+\eps')^{\Theta(q)}$.
Since $q$ might be of order $\Omega(n^{1-\alpha})$,
to make the answers $(1+\eps)$-approximate, we would need to set $\eps'=\Ot(\eps/\poly(n))$
which would, in turn, increase the recomputation time by a polynomial factor.

Instead, we proceed as follows. We maintain the following invariant:
we store a matrix $D_b$, for $b=1,\ldots,\lceil \log_2(n/F)\rceil=b^*$,
such that when phase $j$ starts, $D_b$ is a matrix of $(1+\eps')^{O(b^*-b+1)}$-approximate shortest paths $D$\
of the graph $G$ plus the updates in phases $1,\ldots,l-1$,
where $l$ is the largest integer satisfying $l\leq j-2^{b-1}$ and $2^b\mid l$.

Let us now show how the invariant is maintained.
It is enough to guarantee that for all $b$ such that $2^b\mid j$,
when phase $j+2^{b-1}$ starts, the matrix $D_b$ is reset to an $(1+\eps')^{O(b^*-b+1)}$-approximate shortest paths matrix
capturing the updates in phases $1,\ldots,j-1$.

In the phase $j$, we do the following.
Suppose $b_j$ is the largest integer such that $2^{b_j}\mid j$.
For each $b=\min(b_j,b^*),\ldots,1$, we do the following.
If $b=b^*$, then define $D'$ to be a fresh $(1+\eps')$-approximate shortest paths matrix,
computed using Lemma~\ref{l:appr-recompute} with $G$ equal to an empty graph
and $E^+$ equal to all updates. 
Otherwise, if $b<b^*$,
set $D'$ to be computed from $D_{b+1}$ and $O(F\cdot 2^{b_j})$ updates in the phases
$j-2^{b_j},\ldots,j-1$ using Lemma~\ref{l:appr-recompute}.
Note that by the invariant, and since $2^{b+1}\mid j-2^{b}$,
at this point $D_{b+1}$ stores $(1+\eps')^{O(b^*-b)}$-approximate shortest paths after the updates in phases
$1,\ldots,j-2^{b}-1$.

In both cases, the computation of $D'$ is initiated when phase $j$ starts and is performed in equal time slices per each of the $O(2^{b}\cdot F)$ updates
in the phases $j,\ldots,j+2^{b-1}-1$, so that when
it is ready at the beginning of the phase $j+2^{b-1}$, we substitute
$D'$ for $D_{b}$.
It is also easy to see that $D'$ is $(1+\eps')^{O(b^*-b+1)}$-approximate
since it is obtained via a single application of Lemma~\ref{l:appr-recompute}
to a $(1+\eps')^{O(b^*-b)}$-approximate matrix.

For any $b=1,\ldots,b^*$, we apply Lemma~\ref{l:appr-recompute} to compute $D_b$ with $|W|=O(2^b\cdot n^\alpha)$
once per $2^b\cdot F$ updates.
Since this costs time proportional (up to polylogarithmic factors) to the cost of multiplying
an $n\times (2^b\cdot n^\alpha)$ by a $(2^b\cdot n^\alpha)\times n$ matrix using approximate min-plus product,
and such a task can be reduced to $2^b$ multiplications of $n\times n^\alpha$ and $n^\alpha\times n$ sized matrices,
this takes $\Ot(2^b\cdot n^{\omega(1,\alpha,1)})$ time.
This cost is distributed in equal parts through $2^{b-1}\cdot F$ updates,
so the worst-case per-update cost incurred by this computation is $\Ot(n^{\omega(1,\alpha,1)-\alpha}\cdot (1/\eps)\log(C/\eps))$.
Since $b^*=O(\log{n})$, through all $b$, the
worst-case recomputation cost per update remains
$\Ot(n^{\omega(1,\alpha,1)-\alpha}\cdot (1/\eps)\log(C/\eps))$.

Note that by the invariant, when a phase $j$ starts, $D_1$ contains approximate shortest paths
after updates in phases $1,\ldots,j-2$.
By applying Lemma~\ref{l:appr-recompute} once more to $D_1$ and the updates
in the phase $j-1$, we obtain in $\Ot(n^{\omega(1,\alpha,1)-\alpha}\cdot (1/\eps)\log{(C/\eps)})$
time $(1+\eps')^{\Theta(\log{n})+1}$-approximate shortest paths $D$ after updates
in phases $1,\ldots,j-1$.
Again, to convert this per-phase recomputation cost into a worst-case
per-update cost of $\Ot(n^{\omega(1,\alpha,1)-\alpha}\cdot (1/\eps)\log{(C/\eps)})$,
one proceeds in a standard way as described in Section~\ref{s:incremental-reach}
and also analogously to how the invariant is maintained.

Recall that we denote by $G_0$ the graph at the beginning of a phase
and $k\leq F$ is the number of updates issued in the current phase so far.

\subsection{Path reporting with linear cost or better}\label{s:approx-path}
We first describe a trade-off analogous to that of Section~\ref{s:reach-path}.
Suppose the updates are single-edge updates, i.e., $E_i=\{e_i\}$ for $e_i=u_iv_i$.

For each $u\in V$ and $i=1,\ldots,k$ we store
a $(1+\eps')^{O(\log{n})}$-approximate shortest path $P_{i,u}=u\to u_i$
in $G_0\cup E^+_{i-1}$ (if it exists).
Symmetrically, for each $v\in V$ and $i=1,\ldots,k$, we store
a $(1+\eps')^{O(\log{n})}$-approximate shortest path $Q_{i,v}=v_i\to v$
in $G_0\cup E^+_{i-1}$.

To answer an $s\to t$ shortest path query, we compute the shortest among the paths\linebreak
$\{D_{s,t}\}\cup \{P_{i,s}\cdot e_i\cdot Q_{i,t}:i=1,\ldots,k\}$.
This clearly takes $O(k)=O(n^\alpha)$ time.
The correctness follows easily by the requirements posed on the stored paths.
We stress that the above gives merely the path's length and a pointer to a concatenable deque storing the
underlying path $P$. Actually constructing $P$, which may turn out to be non-simple,
requires iterating through the entire deque in $O(|P|)$ time.
In general, if $C$ is large, this cost may be super-linear in $|V(P)|$ -- as opposed to what
happened in Section~\ref{s:reach-path}.

However, if $C=O(1)$, which happens e.g., for unweighted graphs, we have that
the length of the path can only be at most a constant factor larger than the hop-length
of a path.
As a result, in that case, even if the returned path $P$ is non-simple, $\eps=O(1)$ guarantees that $|P|$ is a
constant factor away from $|P'|$, where $P'$ is obtained from $P$ by eliminating cycles.

To process an update $e_k$, we need to compute the paths $P_{i,u}$ and $Q_{i,v}$ for $u,v\in V$.
We do this using the query procedure in $O(n^{1+\alpha})$ time.
It is crucial to observe that the query procedure is ``exact'' and thus
does not introduce additional multiplicative error.
In other words, if the stored paths $P_{i,\cdot},Q_{i,\cdot}$ were $(1+\eps')^\ell$-approximately shortest
in $G_0\cup E_{i-1}^+$,
the newly computed paths $P_{i,\cdot}$ and $Q_{i,\cdot}$ are 
$(1+\eps')^\ell$-approximately shortest in $G_0\cup E_k^+$ as well.

Since the paths are implemented using concatenable deques of~\cite{KaplanT99} as in Section~\ref{s:reach-path},
the asymptotic cost of computing their representation is the same as the cost
of computing their lengths.

By picking a suitable $\eps'=\eps/\Theta(\log{n})$, and $\alpha=\rho$, we obtain the following lemma.

\begin{lemma}\label{l:approx-path}
  Let $\eps\in (0,1)$.
  There exists a deterministic incremental data structure supporting insertions of single edges
  with real weights in $[1,C]$
  in $\Ot(n^{1+\rho}\cdot (1/\eps)\cdot \log{(C/\eps)})=$\linebreak $O(n^{1.529}\cdot (1/\eps)\cdot \log(C/\eps))$ worst-case time and $(1+\eps)$-approximate shortest path
  queries in \linebreak $O(n^{\rho}+|P|)=O(n^{0.529}+|P|)$
  time, where $|P|$ is the size of the returned \emph{not necessarily} simple path.
\end{lemma}

\begin{remark}
  Lemma~\ref{l:approx-path} appears strictly more general than Lemma~\ref{l:reach-path},
  as it can as well handle reachability queries discussed in Section~\ref{s:reach-path} within the same bound, up to $\Ot(1)$ factors.
  However, the construction behind Lemma~\ref{l:reach-path} allows reporting
  the subsequent edges of a \emph{simple} path with worst-case $O(1)$-time overhead per each edge.
  Lemma~\ref{l:approx-path}, on the other hand, can only report edges
  of a \emph{non-necessarily simple} path  (albeit still approximately shortest in terms of length, and whose hop-length is bounded by its length since the minimum edge weight is $1$) very efficiently.

  It seems that if one wants a simple path to be reported, one needs to first construct
  the entire non-simple path, and only then eliminate cycles.
\end{remark}

\subsection{Path reporting with optimized update time}

In this section we sketch how one can balance the update and the query cost and
obtain a bound polynomially smaller than $O(n^{1+\rho})$.
Again, we consider single-edge updates only.

We subdivide each phase into subphases
of length $f=\Theta(n^\beta)$, where $f$ divides $F$.
Let $G_0'$ be the graph $G$ at the beginning of the current subphase.
Let $X$ be the set of endpoints $u_i$ in the current subphase.
Similarly, let $Y$ be the set of endpoints $v_i$ in the current subphase.
We will guarantee that the approximate shortest path queries regarding $G_0'$
can be answered in $O(n^\alpha)$ time using the data structure
of Section~\ref{s:approx-path}. After the subphase ends,
its edge updates will be ``merged'' into that data structure.

Throughout the subphase, we use a different approach for answering queries.
Namely, suppose that $k'$ updates $e_i=u_iv_i$ for $i=1,\ldots,k'$, were issued in the current subphase.
Then, construct an auxiliary graph $H_{k'}$ on the vertices $X\cup Y$ as follows.
First, the $k'$ edges $e_i$ are inserted to $H_{k'}$.
Next, for each $i=1,\ldots,k'$, $j=1,\ldots,k'$, there is a weighted edge $e_{i,j}=v_iu_j$
in $H_{k'}$ corresponding to an approximate shortest path from $v_i$ to $u_j$ in~$G_0'$ (if such a path
exists in $G_0'$).
The weight of each such edge, along with a pointer to the corresponding underlying path in~$G_0'$, can be computed
in $O(n^{\alpha})$ time by running a query algorithm from Section~\ref{s:approx-path}.

Crucially, observe that $H_{k'}$ can be obtained from $H_{k'-1}$ using $O(k')=O(n^\beta)$
edge insertions, i.e., in $O(n^{\alpha+\beta})$ time.

Now, if a query about an approximate shortest $s\to t$ path is issued, we first
construct a graph $H_{s,t}$ by adding to $H_{k'}$ edges
of the form $st$, $su_i$, $v_it$, for $i=1,\ldots,k'$,
corresponding, again, to approximate shortest paths between their respective endpoints
in $G_0'$ (if they exist).
As a result, $H_{s,t}$ can be obtained from $H_{k'}$ in $O(k'\cdot n^{\alpha})=O(n^{\alpha+\beta})$ time.
Finally, we answer the query by running Dijkstra's algorithm from $s$
on $H_{s,t}$ and returning the found path (lifted to a path in $G$ by concatenating the respective deques) in $O(n^{2\beta})$ time.
The correctness follows by the below lemma, whose easy proof is analogous to that of Lemma~\ref{l:auxiliary-graph}
and is thus omitted.
\begin{lemma}
  The shortest $s\to t$ path in $H_{s,t}$ corresponds to some $(1+\eps')^{O(\log{n})}$-approximately
  shortest path in $G$.
\end{lemma}

Now let us describe how the edges inserted in a subphase are ``merged'' into the data
structure for answering queries in $G_0'$,
so that the state of the data structure is correct before the start of the next subphase.
If we proceeded as in Section~\ref{s:approx-path}, i.e., added these edges one by one,
we would end up spending $O(n^{1+\alpha})$ time per edge, which would not lead to any improvement at all.

Instead, we proceed as follows. Our goal is to compute such paths $P_{i,u}$, $Q_{i,v}$, $u,v\in V$, $i=1,\ldots,k'$, that
$P_{i,u}$ is an approximate shortest $u\to u_i$ path in $G_0'\cup \{u_iv_i:i=1,\ldots,k'\}$,
and $Q_{i,v}$ is an approximate shortest $v_i\to v$ path in that graph.

We first compute the approximate shortest paths between $X\cup Y$ in the graph $G_0'$.
Note that since the query algorithm in Section~\ref{s:approx-path} locates, for a single pair $(x,y)\in X\times Y$,
such an index $i\leq F$ that $P_{i,x}\cdot e_i\cdot Q_{i,y}$ is minimized,
doing it for all pairs $(x,y)$ is precisely computing an (exact) min-plus distance product.
We do not however need to compute it exactly; computing it approximately is sufficient
if we guarantee that not too much error will accumulate.
Let us forget, for a moment, about this issue.
Then, we can find approximate shortest paths matrix $R$ between the vertices $X\cup Y$ in $G_0'$
using Theorem~\ref{t:zwick-appr} in $\Ot(n^{\omega(\beta,\alpha,\beta)}\cdot (1/\eps)\cdot \log{(C/\eps)})$ time.
The matrix $R$ can be lifted to approximate shortest paths matrix $R^*$ between $X\cup Y$ in $G_0'\cup \{u_iv_i:i=1,\ldots,k'\}$
by repeated squaring of the matrix $R'$, where $R'_{x,y}$ is set to the shortest
out of the path $R_{x,y}$ and a single-edge path $xy$, if $xy=u_iv_i$ for some $i$
(this is completely analogous to the computation performed in the proof of Lemma~\ref{l:appr-recompute}).
Computing $R^*$ again takes $\Ot(n^{\omega(\beta,\beta,\beta)}\cdot (1/\eps)\cdot \log{(C/\eps)})$ time.

Next, we similarly use Lemma~\ref{l:appr-recompute} to compute approximate shortest paths matrices $P',Q'$ between pairs
in $(V\times (X\cup Y))\cup ((X\cup Y)\times V)$ in the graph $G_0'$; this takes\linebreak
$\Ot(n^{\omega(1,\alpha,\beta)}\cdot (1/\eps)\cdot \log{(C/\eps)})$ time
using the approximate distance product.
We then set $P'_{u_i,v_i}=\min(P'_{u_i,v_i},\wei_G(u_iv_i))$
and $Q'_{u_i,v_i}=\min(Q'_{u_i,v_i},\wei_G(u_iv_i))$.
Finally, to compute the approximate shortest paths $P_{i,u}$ in
$G_0'\cup \{u_iv_i:i=1,\ldots,k'\}$, it is enough to compute one more approximate distance
product of $P'$ and $R^*$.
Similarly, to compute approximate shortest paths $Q_{i,u}$,
it is enough to compute one more approximate distance product of $R^*$ and $Q'$.
Both these computations require \linebreak
$\Ot(n^{\omega(1,\beta,\beta)}\cdot (1/\eps)\cdot \log{(C/\eps)})$ time.

To conclude, the new subphase's updates can be merged into the data structure
answering queries in $\Ot(n^{\omega(1,\alpha,\beta)}\cdot (1/\eps)\cdot \log{(C/\eps)})$ time,
which is a cost incurred once per a subphase, i.e., the amortized time
spent on this is  $\Ot(n^{\omega(1,\alpha,\beta)-\beta}\cdot (1/\eps)\cdot \log{(C/\eps)})$.
This is however, a slight oversimplification:
to make this bound hold in the worst-case, and make sure that the error does not
accumulate too much (i.e., be as large as $(1+\eps')^{\Theta(F/f)}=(1+\eps')^{\poly(n)}$), we need to proceed similarly as in Section~\ref{s:approx-path} with matrices $D_b$:
merges of sizes $f,2f,4f,\ldots,2^i\cdot f,\ldots,\Theta(F)$ have
to be performed every $\Theta(2^i\cdot f)$ updates, so that the accumulated multiplicative
error stays $(1+\eps')^{\Theta(\log{n})}$.

The worst-case update time is $\Ot((n^{\omega(1,\alpha,1)-\alpha}+n^{\omega(1,\alpha,\beta)-\beta}+n^{\alpha+\beta})\cdot (1/\eps)\cdot \log(C/\eps))$,
whereas the query time for returning a representation of a $(1+\eps)$-approximate shortest path is $O(n^{\alpha+\beta})$.
This is the same trade-off as in~\cite[Theorem~4.2]{BrandNS19}, and with the currently
known bounds on $\omega(\cdot,\cdot,\cdot)$, one optimizes it
by setting $\alpha\approx0.855$ and $\beta\approx0.551$.
We thus obtain the following.

\begin{lemma}\label{l:approx-path-opt}
  Let $\eps\in (0,1)$.
  There exists a deterministic incremental data structure supporting insertions of single edges
  with real weights in $[1,C]$ in $O(n^{1.407}\cdot (1/\eps)\cdot \log(C/\eps))$
  worst-case time and $(1+\eps)$-approximate shortest path queries
  in $O(n^{1.407}+|P|)$ time, where $|P|$ is the size of the returned \emph{not necessarily} simple path.
\end{lemma}

Note that the $|P|$ term in the query time above can be neglected in the case
of unweighted graphs and $\eps=O(1)$, which can be used to obtain a deterministic incremental
reachability data structure supporting both single-edge updates and path queries
in $O(n^{1.407})$ worst-case time.

Lemmas~\ref{l:approx-path}~and~\ref{l:approx-path-opt} together yield Theorem~\ref{t:incremental-appr-sum}.

\section{Incremental exact shortest paths}\label{s:incremental-exact}
In this section we work with incremental unweighted graphs and our goal is to support
computing shortest paths exactly.
We will consider two settings:
\begin{itemize}
  \item single-edge updates and single-pair shortest path queries,
  \item single-vertex incoming edges updates and single-source shortest paths tree queries.
\end{itemize}

The data structures maintained by the algorithm will be the same
for both settings. We will use the following analogue of Lemma~\ref{l:trees-recompute}.

\begin{lemma}\label{l:exact-recompute}
  Let $h\leq n$ be an integer.
  Suppose the matrix $D$ represents exact all-pairs distances in~$G$ for pairs
  of vertices whose distance does not exceed the threshold $h$.
  Let the matrix $P$ contain the respective underlying paths, implemented as concatenable deques as in Section~\ref{s:reach-path}.
  Let~$W\subseteq V$ be of size $O(n^\alpha)$, where $\alpha\in [0,1]$.
  Let $E^+$ be some set of edges with heads in the set $W$.

  One can construct an analogous matrix $D'$ representing exact
  distances that do not exceed~$h$ in $G\cup E^+$, along with a matrix of respective simple paths $P'$,
  deterministically in $\Ot(h\cdot n^{\omega(1,\alpha,1)})$ time.
\end{lemma}
\begin{proof}[Proof sketch.]
  The proof is almost identical to that of Lemma~\ref{l:appr-recompute}, so we merely sketch it.
  One only needs to substitute every application of approximate min-plus product $\star_\eps$~(Theorem~\ref{t:zwick-appr})
  with exact min-plus product $\star_{\leq h}$ of matrices whose finite entries are positive
  integers not larger than~$h$~\cite[Lemma~3.3]{Zwick02}.
  This bounded min-plus product can be computed in $\Ot(hn^{\omega(1,\alpha,1)})$ time
  and can provide the corresponding witnesses deterministically within the same bound.
\end{proof}

The algorithm will again operate in phases of insertions of length $F=n^\alpha$.
Let again $G_0$ denote the graph at the beginning of the current phase, and let
$E_i$ be the edges inserted in the $i$-th update of the phase.
Denote by $v_i$ the head of the edges inserted in the $i$-th update.
If just a single edge is inserted in that update, denote by $u_i$ the tail of that edge.
Let $k$ denote the most recent update's number within the current phase and let $E_k^+=\bigcup_{i=1}^k E_i$.

Whenever a phase starts, we will recompute a matrix of shortest paths $P$ up to length $h=n/F$ using
Lemma~\ref{l:exact-recompute} in $\Ot(hn^{\omega(1,\alpha,1)})=\Ot(n^{\omega(1,\alpha,1)+1-\alpha})$ time,
based on the matrix from the beginning of the previous phase.
The next step is to compute a small hitting set $B\subseteq V$ of all paths in $P$
of length \emph{precisely} $n/F$.
Such a hitting set of size $|B|=\Ot(F)$ can be computed deterministically
using a greedy algorithm in time linear in the input size, i.e., in $\Ot(n^{3-\alpha})=\Ot(n^{\omega(1,\alpha,1)+1-\alpha})$ time in our case (see, e.g.,~\cite{Zwick02}).
So the amortized time spent on recomputation
is $\Ot(n^{\omega(1,\alpha,1)+1-2\alpha})$.
This can be turned into a worst-case bound in a standard way as was already done
multiple times in the previous sections.
We will use the following trivial observation.
\begin{observation}\label{o:hitting}
  Suppose $B$ is a hitting set of all shortest paths of length precisely $h$ in $G_0$.
  Then, $B\cup \{v_i:i=1,\ldots,k\}$ is a hitting set of all shortest paths of length
  precisely $h$ in $G_0\cup E_k^+$.
\end{observation}

\subsection{Single-edge updates, pair queries}

Let $X=\{u_i,v_i:i=1,\ldots,k\}$ denote the set of endpoints of edge insertions in the current phase.
We have $|X|=O(F)=O(n^\alpha)$.
We construct an auxiliary graph $H_{s,t}$ with vertices $\{s,t\}\cup X\cup B$.
The edge set of $H_{s,t}$ includes the edges $E_k^+$ and a weighted edge $uv$ for each $u,v\in V(H_{s,t})$
corresponding to the path $P_{u,v}$ in $G_0$, if $P_{u,v}$ has length no more than $h$.
Finally, we compute an $s\to t$ shortest path in $H_{s,t}$ using
Dijkstra's algorithm in $O(|X\cup B|^2)=\Ot(F^2)=\Ot(n^{2\alpha})$ time.
Clearly, such a path in $H_{s,t}$ can be lifted to a path
in $G=G_0\cup E^+_k$ using the paths stored in the matrix $P$.
Note that since the graph is unweighted, the size of the returned path is $O(n)$.

The correctness follows easily by Observation~\ref{o:hitting} and the fact that
all the relevant paths of length no more than $h$ in $G_0$ are represented in $H_{s,t}$.

\begin{lemma}\label{l:exact-path}
  Let $G$ be an unweighted digraph and let $\alpha\in (0.5,1)$.
  There exists a deterministic incremental data structure supporting single-edge insertions
  in $\Ot(n^{\omega(1,\alpha,1)+1-2\alpha})$ worst-case time~and queries reporting a shortest path from
  a given source to a given target in $\Ot(n^{2\alpha})$ time.
  In particular, for $\alpha=0.81$, both bounds are $O(n^{1.62})$.
\end{lemma}

\begin{remark}
The above lemma holds even for $\alpha\leq 0.5$, but then the update time is $\Omega(n^2)$.
However, $O(n^2)$ worst-case update time with $O(1)$ query time can be obtained rather trivially even for real-weighted graphs,
  so the data structure behind Lemma~\ref{l:exact-path} is inferior in that case.
\end{remark}

\subsection{Incoming edges updates, single-source queries}

We construct an auxiliary graph $H_s$ with vertices $V$. Let $Y=\{v_i:i=1,\ldots,k\}$.
The edge
set of~$H_s$ includes the edges $E_k^+$ (recall that $|E_k^+|\leq nk=O(n^{1+\alpha})$).
Moreover, for each $v\in B\cup Y\cup\{s\}$ and $w\in V$, we add a weighted edge
$vw$ representing the path $P_{v,w}$ in $G_0$, unless $P_{v,w}$ has length more than $h$.
As a result, the graph $H_s$ has $\Ot(nF)=\Ot(n^{1+\alpha})$ edges.

By proceeding similarly as in the proof of Lemma~\ref{l:auxiliary-graph},
it is also easy to show that $H_s$ preserves distances from $s$ in $G\cup E_k^+$.
All the required paths of length no more than $h$ are represented in $H_s$
by the definition of the hitting set $B$ and Observation~\ref{o:hitting}.

A shortest paths tree $T_s$ from $s$ in $H_s$ can be computed in $\Ot(n^{1+\alpha})$ time
using Dijkstra's algorithm.
Note that each edge in that tree corresponds to either a path of length $\leq h=O(n^{1-\alpha})$
in $G_0$, or a single edge in $E_k^+$. 
As a result, by expanding the $O(n)$ edges of $T_s$
into paths in $G_0$, we obtain a subgraph $T_s'\subseteq G$ with $\Ot(n^{2-\alpha})$ edges
that is guaranteed to contain some shortest $s\to t$ path in $G$ for any $t$.
Consequently, a shortest paths tree from $s$ in $G$ can be obtained by running
Dijkstra's algorithm once again on $T_s'$ in $\Ot(n^{2-\alpha})$ time.
For $\alpha\geq 0.5$ this is dominated by the $\Ot(n^{1+\alpha})$ bound.
Therefore, we obtain the following.

\begin{lemma}\label{l:exact-tree}
  Let $G$ be an unweighted digraph and let $\alpha\in (0.5,1)$.
  There exists a deterministic incremental data structure supporting insertions of incoming edges of a single vertex
  in \linebreak $\Ot(n^{\omega(1,\alpha,1)+1-2\alpha})$ worst-case time
  and queries reporting a single-source shortest paths tree from a given source in $\Ot(n^{1+\alpha})$ time.
  In particular, for $\alpha=0.724$, both bounds are $O(n^{1.724})$.
\end{lemma}

Lemmas~\ref{l:exact-path}~and~\ref{l:exact-tree} together yield Theorem~\ref{t:incremental-exact-sum}.

\bibliographystyle{alpha}
\bibliography{references}

\end{document}